\documentclass[runningheads]{llncs}
\usepackage[T1]{fontenc}
\usepackage{booktabs} 
\usepackage{hyperref,amsmath,amssymb}
\usepackage[ruled,vlined,linesnumbered,noend]{algorithm2e}
\SetAlFnt{\small}
\SetAlCapFnt{\small}
\SetAlCapNameFnt{\small}
\SetAlCapHSkip{0pt}

\SetCommentSty{myalgcommfnt}
\SetKwInput{Initial}{Initial call}
\SetKwProg{Fn}{}{}{}

\usepackage{proof}
\usepackage{float}
\usepackage{xspace} 
\usepackage{xparse}
\usepackage{wrapfig}
\usepackage{paralist}
\usepackage[shortlabels]{enumitem}
\usepackage[capitalize]{cleveref}

\usepackage{tikz}
\usetikzlibrary{automata,tikzmark,decorations.pathreplacing,petri,backgrounds,positioning,calc,trees,arrows,shapes}

\usepackage{pgfplots}
\pgfplotsset{compat=1.13}
\usepackage{pgfplotstable}
\usepackage[nointegrals]{wasysym} 

\usepackage{microtype}
\usepackage{graphicx}

\usepackage{lineno}

\newcommand{\bjnote}[1]{}
\newcommand{\bjforget}[1]{}
\newcommand{\hasbeenremoved}[1]{}

\newcommand\EventDPOR{\mbox{Event-DPOR}\xspace}
\newcommand\OptimalDPOR{\mbox{Optimal-DPOR}\xspace}
\newcommand\LAPOR{\mbox{LAPOR}\xspace}

\newcommand\CDSChecker{\textsc{CDSChecker}\xspace}
\newcommand\GenMC{\textsc{GenMC}\xspace}
\newcommand\RCMC{\textsc{RCMC}\xspace}
\newcommand\Nidhugg{\textsc{Nidhugg}\xspace}

\newcommand{\explore}{\ensuremath{Explore}}

\newcommand{\done}{\ensuremath{done}\xspace}
\newcommand{\accesses}{\ensuremath{msgAccesses}}
\newcommand{\tmpaccesses}{\ensuremath{tmpAccesses}}

\newcommand{\pendingwus}[1]{\ensuremath{\pendingwusname(#1)}}
\newcommand{\pendingwusname}{\mathit{parkedWuS}}

\newcommand{\insertpendingwu}[2]{\ensuremath{\mbox{\it InsertParkedWuS}(#1,#2)}}
\newcommand{\insertpendingwuname}{\ensuremath{\mbox{\it InsertParkedWuS}}\xspace}
\newcommand{\reverserace}{\ensuremath{ReverseRace}}
\newcommand{\WuT}{\ensuremath{WuT}}

\newcommand{\treeorder}{\ensuremath{\propto}}

\newcommand{\event}{e}

\newcommand{\msgof}[2]{\transpair{#1}{#2}}

\newcommand{\weaksatrel}[1]{\ensuremath{\langle\!\langle{#1}\rangle\!\rangle}}

\newcommand{\hbwpref}[2]{\ensuremath{\xrightarrow{hbp\langle{#1}\rangle}_{#1.#2}}}

\newcommand{\exseq}{E}

\newcommand{\set}[1]{\{#1\}}

\newcommand{\tuple}[1]{\ensuremath{\langle {#1}\rangle}}
\newcommand{\happbf}[2]{\ensuremath{\xrightarrow{#1}_{#2}}}
\newcommand{\nhappbf}[2]{\ensuremath{\overset{#1}{\nrightarrow_{#2}}}}

\newcommand{\dom}[1]{\ensuremath{dom({#1})}}

\newcommand{\prefix}{\leq}
\newcommand{\totorder}[1]{<_{#1}}

\newcommand{\revrace}[1]{\ensuremath{\lesssim_{#1}}}
\newcommand{\revmsgrace}[2]{\ensuremath{\lesssim_{#1}^{#2}}}

\newcommand{\wut}[1]{\ensuremath{wut({#1})}}

\newcommand{\enabled}[1]{\ensuremath{enabled({#1})}}

\newcommand{\exseqs}{\ensuremath{\mathcal{E}}}

\newcommand{\winits}[2]{\ensuremath{WI_{[{#1}]}({#2})}}

\newcommand{\infirstseqs}[1]{\ensuremath{\sqsubseteq_{[#1]}}}

\newcommand{\remove}{\! \setminus \!}

\newcommand{\insertwus}[3]{\mathit{Insert({#1},{#2},{#3})}}

\newcommand{\insertwusname}{\ensuremath{\mathit{InsertWuS}}\xspace}

\newcommand{\mtequiv}{\simeq}
\newcommand{\equivafter}[1]{\simeq_{[#1]}}
\newcommand{\mtprefix}{\sqsubseteq}

\newcommand{\mtprefixafter}[1]{\mtprefix_{[#1]}}
\newcommand{\notmtprefixafter}[1]{\ \ / \!\!\!\! \mtprefix_{[#1]}}

\newcommand{\eqclass}[1]{\ensuremath{[{#1}]_\simeq}}
\newcommand{\emptyseq}{\ensuremath{\langle\rangle}}
\newcommand{\emptytree}{\ensuremath{\langle\{\langle\rangle\},\emptyset\rangle}}
\newcommand{\nextev}[2]{\ensuremath{next_{[{#1}]}({#2})}}
\newcommand{\mtclass}[1]{[#1]_{\simeq}}
\newcommand{\valid}[2]{\ensuremath{{#1}\!\vdash\!{#2}}}
\newcommand{\subtreeafter}[2]{\mbox{\sl subtree}(#2,#1)}

\newcommand{\post}{\texttt{post}}

\newcommand{\po}{\ensuremath{\color{red}\mathtt{po}\color{black}}}

\newcommand{\rf}{\ensuremath{\color{green}\mathtt{rf}\color{black}}}

\newcommand{\pb}{\ensuremath{\color{violet}\mathtt{pb}\color{black}}}

\newcommand{\cnf}{\ensuremath{\color{red}\mathtt{cnf}\color{black}}}

\newcommand{\hb}{\ensuremath{\color{blue}\mathtt{hb}\color{black}}}

\newcommand{\sat}{\ensuremath{\color{gray}\mathtt{st}\color{black}}}

\newcommand{\procof}[1]{\widehat{#1}}
\newcommand{\pair}[2]{\langle#1,#2\rangle}
\newcommand{\transpair}[2]{\pair{#1}{#2}}

\NewDocumentCommand{\rtrans}{oo} 
{\IfNoValueTF {#1}
        {\ensuremath{\xrightarrow{}}\xspace}
        {\IfNoValueTF {#2}
            {\ensuremath{\xrightarrow{#1}}\xspace}
            {\ensuremath{\xrightarrow[#2]{#1}}\xspace}}}

\NewDocumentCommand{\rntrans}{oo} 
{\IfNoValueTF {#1}
        {\ensuremath{\centernot{\xrightarrow{}}}\xspace}
        {\IfNoValueTF {#2}
            {\ensuremath{\centernot{\xrightarrow{#1}}}\xspace}
            {\ensuremath{\centernot{\xrightarrow[#2]{#1}}}\xspace}}}

\newcommand{\keyword}[1]{\mbox{\bf #1}}

\SetKw{Continue}{continue}

\tikzset{event/.style={draw=none,fill=none,align=left}}
\tikzset{arc/.style={->,>=stealth,thick}}
\tikzset{arcDerived/.style={->,>=stealth,thick,dashed}}

\newcommand{\bench}[1]{\textsf{\mbox{#1}}}
\newcommand{\evt}[1]{\texttt{#1--event}\xspace}
\newcommand{\opt}[1]{\texttt{#1--optimal}\xspace}
\newcommand{\genmcmo}[1]{\texttt{#1--mo}\xspace}
\newcommand{\lapormo}[1]{\texttt{#1--lapor}\xspace}

\newlist{propenum}{enumerate}{2} 
\setlist[propenum]{ref=\arabic*, label={\arabic*.}}
\crefname{propenumi}{property}{properties}
\crefname{propenumii}{property}{properties}

\newlist{ruleenum}{enumerate}{2} 
\setlist[ruleenum]{ref=R\arabic*, label={R\arabic*.}}
\crefname{ruleenumi}{rule}{rules}
\Crefname{ruleenumi}{Rule}{Rules}

\newlist{phaseenum}{enumerate}{2} 
\setlist[phaseenum]{ref=P-\Roman*, label={P-\Roman*.}}
\crefname{phaseenumi}{phase}{phases}
\Crefname{phaseenumi}{Phase}{Phases}

\makeatletter
\newcommand{\iffref}[3]{\@ifundefined{r@#1}{#3}{#2}}
\makeatother

\newcommand{\citet}[1]{\cite{#1}}

\begin{document}
\title{Tailoring Stateless Model Checking for Event-Driven Multi-Threaded Programs}
\titlerunning{Tailoring Stateless Model Checking for Event-Driven Programs}
%
\author{Parosh Aziz Abdulla\inst{1} \and Mohammed Faouzi Atig\inst{1}
		\and Frederik Meyer Bønneland\inst{2} \and Sarbojit Das\inst{1}
		\and Bengt Jonsson\inst{1} \and Magnus Lång\inst{1}
		\and Konstantinos Sagonas\inst{1,3}}
\institute{Uppsala University, Uppsala, Sweden \and Aalborg University, Denmark \and NTUA, Grece
}
\authorrunning{P. Abdulla \and M. Atig \and F. Bønneland \and S. Das \and B. Jonsson \and M. Lång \and K. Sagonas}

%
\maketitle              
\begin{abstract}
Event-driven multi-threaded programming is an important idiom for structuring concurrent computations.
Stateless Model Checking (SMC) is an effective verification technique for multi-threaded programs, especially when coupled with Dynamic Partial Order Reduction (DPOR).
Existing SMC techniques are often ineffective in handling event-driven programs,
since they will typically explore all possible orderings of event processing, even when events do not conflict.
We present \EventDPOR, a DPOR algorithm tailored to event-driven multi-threaded programs.
It is based on \OptimalDPOR, an optimal DPOR algorithm for multi-threaded programs; we show how it can be extended for event-driven programs.
We prove correctness of \EventDPOR for all programs, and optimality for a large subclass.
One complication is that an operation in \EventDPOR, which checks for redundancy of new executions, is NP-hard,
as we show in this paper; we address this by a sequence of inexpensive (but incomplete) tests which check for redundancy efficiently.
Our implementation and experimental evaluation show that, in comparison with other tools in which handler threads are simulated using locks, \EventDPOR can be exponentially faster than other state-of-the-art DPOR algorithms on a variety of programs and

\end{abstract}
%


\section{Introduction}
\label{sec:intro}
Event-driven multi-threaded programming is an important idiom for structuring concurrent computations
in distributed message-passing applications, file systems~\cite{Mazieres01},
high-performance servers~\cite{Dabek:event-driven-02}, systems programming~\cite{P:pldi13},
smartphone applications~\cite{mednieks2012programming}, and many other domains.
In this idiom, multiple threads execute concurrently and can communicate through shared objects.
In addition, some threads, called \emph{handler threads}, have an associated event pool
to which all threads can post events.
Each handler thread executes an event processing loop in which events
from its pool are processed sequentially, one after the other, interleaved with the execution of other threads.
An event is processed by invoking an appropriate handler, which can be, e.g., a callback function.


Testing and verification of event-driven multi-threaded programming faces
all the usual challenges of testing and verification for multi-threaded programs, and furthermore suffers from
additional complexity, since the order of event execution is determined dynamically and non-deterministically.
A successful and fully automatic technique for finding concurrency bugs in multithreaded programs (i.e., defects that
arise only under some thread schedulings) and for verifying their absence is \emph{stateless model checking} (SMC)~\cite{Godefroid:popl97}.
Given a terminating program and fixed input data,
SMC systematically explores the set of
all thread schedulings that are possible during program runs.
A special runtime scheduler drives the SMC exploration by making decisions
on scheduling whenever such choices may affect the interaction between threads.
SMC has been implemented in many tools
(e.g., VeriSoft~\cite{Godefroid:verisoft-journal},
\textsc{Chess}~\cite{MQBBNN:chess}, Concuerror~\cite{Concuerror:ICST13},
\Nidhugg~\cite{tacas15:tso}, rInspect~\cite{DBLP:conf/pldi/ZhangKW15},
\CDSChecker~\cite{NoDe:toplas16}, \RCMC~\cite{KLSV:popl18}, and
\GenMC~\cite{GenMC@CAV-21}), and successfully applied to realistic
programs (e.g.,~\cite{GoHaJa:heartbeat} and~\cite{KoSa:spin17}).
To reduce the number of explored executions,
SMC tools typically employ \emph{dynamic partial order reduction}
(DPOR)~\cite{FG:dpor,abdulla2014optimal}.
DPOR defines an equivalence relation on executions, which
preserves relevant correctness properties, such as reachability of local
states and assertion violations, and explores at least one execution in each equivalence class.
\hasbeenremoved{We call a DPOR algorithm \emph{optimal} if it guarantees the exploration of exactly one execution per
equivalence class.}


Existing DPOR techniques for multi-threaded programs lack effectiveness in handling the complications brought by event-driven programming,
as has been observed by e.g., Jensen et al.~\cite{Event-DrivenSMC@OOPSLA-15}
and Maiya et al.~\cite{Maiya:tacas16}.
A na\"ive way to handle such a program is to consider all pairs of events as conflicting,
implying that different orderings of event executions by a handler thread will be considered inequivalent.
A major drawback is then that a DPOR algorithm cannot exploit the fact that different orderings of event executions by a single handler thread
can be considered equivalent in the case that events are non-conflicting.
In this way, a program in which $n$ non-conflicting events
are posted to a handler thread by $n$ concurrent threads can give rise to $n!$ explorations by a standard
DPOR algorithm, whereas all of them are in fact equivalent.
On the other hand, some events may be conflicting, so a DPOR algorithm for event-driven programs
should explore only the necessary inequivalent orderings between conflicting events.
This can be achieved by defining an equivalence on executions, which respects
only the ordering of conflicting accesses to shared variables,
irrespective of the order in which events are executed.
For plain multi-threaded programs, this equivalence is the basis for several effective DPOR
algorithms~\cite{FG:dpor,abdulla2014optimal}.
The challenge is to develop an effective DPOR algorithm also for event-driven programs.

In this paper, we present \EventDPOR, a DPOR algorithm for event-driven multi-threaded programs
where handlers can execute events from their event pool in arbitrary order (i.e., the event pool is viewed as a multiset).
The multiset semantics is used in many works~\cite{popl07:JhalaM,Raychev:oopsla13,Event-DrivenSMC@OOPSLA-15}, often with the significant restriction that there is only one handler thread; we consider the more general situation with an arbitrary number of handler threads.
\EventDPOR is based on \OptimalDPOR~\cite{abdulla2014optimal,optimal-dpor-jacm},
a DPOR algorithm for multi-threaded programs.
The basic working mode of \OptimalDPOR is similar to several other DPOR algorithms:
Given a terminating program, one of its executions is explored and then analyzed to construct initial fragments of
new executions; each fragment that is not redundant (i.e., which can be extended to an execution that is not equivalent to a previously explored execution),
is subsequently extended to a maximal execution, which is analyzed to construct
initial fragments of new executions, and so on.
\EventDPOR employs the same basic mode of operation as \OptimalDPOR,
but must be extended to cope with the event-driven execution model.
One complication is that the constructed initial fragments must satisfy the constraints imposed
by the fact that event executions on a handler are serialized; this may necessitate
reordering of several events when constructing new executions from an already explored one.
Another complication is that the check whether a new fragment is redundant
is NP-hard in the event-driven setting, as we prove in this paper.
We alleviate this by defining a sequence of inexpensive but incomplete rendundancy checks, using
a complete decision procedure only as a last resort.

We prove that the \EventDPOR algorithm is \emph{correct} (explores at least one execution in each equivalence class) for event-driven programs.
We also prove that it is \emph{optimal} (explores exactly one execution in each equivalence class) for the class of so-called \emph{non-branching} programs, in which the possible sequences of shared variable accesses that can be performed 
during execution of an event, whose handler also executes other events, does  not depend on how its execution is interleaved with other threads.

We have implemented \EventDPOR in an extension of the \Nidhugg
tool~\cite{tacas15:tso}.
Our experimental evaluation shows that, when compared with other SMC tools in
which event handlers are simulated using locks, \EventDPOR incurs only a
moderate constant overhead, but can be exponentially faster than other
state-of-the-art DPOR algorithms. The same evaluation also shows that, unlike
other algorithms that can achieve analogous reduction, \EventDPOR manages to
completely avoid unnecessary exploration of executions that cannot be serialized.
Moreover, in all the programs we tried, 
also those that are not non-branching, \EventDPOR explored the optimal number of traces, suggesting that
\EventDPOR is optimal not only for non-branching programs but also for a good number of branching ones.
Also, our sequence of inexpensive checks for redundancy was sufficient in all tried programs, i.e.,
we never had to invoke the decision procedure for this NP-hard problem.

\section{Related Work} \label{sec:related}

Stateless model checking has been implemented in many tools for analysis of multithreaded programs (e.g., \cite{Godefroid:verisoft-journal,MQBBNN:chess,Concuerror:ICST13,tacas15:tso,DBLP:conf/pldi/ZhangKW15,NoDe:toplas16,KLSV:popl18,GenMC@CAV-21}). It often employs DPOR, introduced by Flanagan and Godefroid~\cite{FG:dpor} to reduce the number of schedulings that must be explored. Further developments of DPOR reduce this number further, by being optimal (i.e., exploring only one scheduling in each equivalence class)~\cite{abdulla2014optimal,optimal-dpor-jacm,observers,KMGV:popl22} or by
weakening the equivalence~\cite{observers,CS-DPOR@CAV-17,DC-DPOR@POPL-18,rfsc@OOPSLA-19}.

DPOR has been adapted to event-driven multi-threaded programs. Jensen et al.~\citet{Event-DrivenSMC@OOPSLA-15} consider an execution model in which events are processed in arbitrary order (multiset semantics) and apply it to JavaScript programs.
Maiya et al.~\citet{Maiya:tacas16} consider a model where events are processed in the order they are received (FIFO semantics), and develop a tool, EM-Explorer, for analyzing Android applications which, given a particular sequence of event executions, produces a set of reorderings of its events which reverses its conflicts.
The above works are based on the algorithm of Flanagan and Godefroid~\citet{FG:dpor}, implying that they do not take advantage of subsequent improvements in DPOR algorithms~\cite{abdulla2014optimal,optimal-dpor-jacm,KMGV:popl22}, nor do they employ techniques such as sleep sets for avoiding redundant explorations. It is known~\cite{optimal-dpor-jacm}that even with sleep sets, the algorithm of Flanagan and Godefroid~\citet{FG:dpor} can explore an exponential number of redundant execution compared to the algorithms of~\cite{abdulla2014optimal,optimal-dpor-jacm,KMGV:popl22}. Without sleep sets, the amount of redundant exploration will increase further.
Recently, Trimananda et al.~\citet{Trimananda:vmcai22} have proposed an adaptation of stateful DPOR~\cite{YWY:stateful-dpor,YangCGK@SPIN-08} to non-terminating event-driven programs, which has been implemented in Java PathFinder. For analogous reason as for~\citet{Event-DrivenSMC@OOPSLA-15,Maiya:tacas16}, also this approach does not avoid to perform redundant explorations.
  
For actor-based programs, in which processes communicate by message-passing, Aronis et al.~\citet{observers} have presented an improvement of \OptimalDPOR in which two postings of messages to a mailbox are considered as conflicting only if their order affects the subsequent behavior of the receiver. Better reduction can then be achieved if the receiver selects messages from its mailbox based on some criterion, such as by pattern matching on the structure of the message.
However, this execution model differs from the one we consider. 

Event-driven programs where handlers select messages in arbitrary order from their mailbox can be analyzed by modeling messages (mini-)threads that compete for handler threads by taking locks, and applying any SMC algorithm for shared-variable programs with locks.
Since typical SMC algorithms always consider different lock-protected code sections as conflicting, this approach has the drawback of exploring all possible orderings of events on a handler.
There exists a technique to avoid exploring of all these orderings in programs with locks, in which lock sections can be considered non-conflicting if they do not perform conflicting accesses to shared variables. This LAPOR technique~\citet{LAPOR@OOPSLA-19} is based on optimistically executing lock-protected code regions in parallel, and aborting executions in which lock-protected regions cannot be serialized. This can led to significant useless exploration, as also shown in our evaluation in \cref{sec:eval}.

The problem of detecting potentially harmful data races in single executions of event-driven programs has been addressed by several works.
The main challenge for data race detection is to capture the often hidden dependencies for applications on Android~\cite{Hsiao:pldi14,Maiya:pldi14,Bielik:android-races,Hu:issta16} or on other platforms~\cite{Petrov:pldi12,Raychev:oopsla13,Santhiar:issta16,Maiya:issta17}.
Detecting data races is a different problem than exploring all possible executions of a program, in that it considers only one (possibly long) execution, but tries to detect whether it (or some other similar execution) exhibits data races.


\newcommand{\tikzwrapfigbg}{%
  \begin{pgfonlayer}{background}
    \path[fill=gray!10,rounded corners]
    (current bounding box.south west) rectangle
    (current bounding box.north east);
\end{pgfonlayer}}

\section{Main Concepts and Challenges}
\label{sec:concepts}
In this section, we informally present core concepts of our approach by examples\footnote{Note
  that in the remainder of the paper, we will use the term \emph{message} to refer to what was called \emph{event} in Sections~\ref{sec:intro} and~\ref{sec:related},
  for the reason that the literature on DPOR has reserved the term \emph{event} to denote an execution of a program statement. We will also use \emph{mailbox} instead of event pool.}

\newcommand{\evnt}[2]{\mbox{$#1$: \texttt{#2}}}
\newcommand{\hndlr}[3]{\mbox{$#1$: $#2$: \texttt{#3}}}

\begin{figure}
  \centering
  \begin{minipage}[b]{0.35\textwidth}
    \centering \footnotesize
    Writer-readers program.
    \begin{tikzpicture}[line width=1pt,framed,inner sep=1pt]
      \node[name=p,anchor=south west] at (-0.15,0.25) {{$s$}};
      \node[name=n11,anchor=mid] at (0,0) {\texttt{x\,=\,1}};
      \draw[line width=0.5pt] ($(n11.north east)+(4pt,10pt)$)--($(n11.south east)+(4pt,-12pt)$);
      \draw[line width=0.5pt] ($(n11.north east)+(6pt,10pt)$)--($(n11.south east)+(6pt,-12pt)$);
      
      \node[name=q,anchor=south west] at (0.9,0.25) {{$t$}};
      \node[name=n21,anchor=west] at ($(n11.east)+(10pt,-1pt)$) {\texttt{a\,=\,y;}};
      \node[name=n22,anchor=north west] at ($(n21.south west)+(0pt,-0.25pt)$) {\texttt{b\,=\,x}};

      \draw[line width=0.5pt] ($(n21.north east)+(5pt,10pt)$) -- ($(n21.south east)+(5pt,-12pt)$);
      \draw[line width=0.5pt] ($(n21.north east)+(7pt,10pt)$) -- ($(n21.south east)+(7pt,-12pt)$);

      \node[name=t,anchor=south west] at (1.95,0.25) {{$u$}};
      \node[name=n31,anchor=west] at ($(n21.east)+(10pt,0.5pt)$) {\texttt{c\,=\,z;}};
      \node[name=n32,anchor=north west] at ($(n31.south west)+(0pt,-2pt)$) {\texttt{d\,=\,x}};
    \end{tikzpicture}
  \end{minipage}
  \begin{minipage}[b]{0.55\textwidth}
    \centering \footnotesize
    \begin{tikzpicture}[
        yscale=.6,xscale=.75,
        level distance=1cm,
        st/.style={ellipse,inner sep=0.2em,draw},
        wu/.style={ellipse,inner sep=0.2em,color=blue,draw},
        old/.style={ellipse,inner sep=0.2em,color=gray,draw},
        fut/.style={ellipse,inner sep=0.2em,color=darkgray,draw},
        child anchor=north
      ]
      \node[st] {}
      child { node[old] {}
        child { node[old] {}
          child { node[old] {}
            child { node[old] {}
              child { node[old,label=south:$E_1$] {}
                edge from parent [->,color=black] node [left] (dx) {\evnt{u}{d\,=\,x}}
              }
              edge from parent [->,color=black] node [left] {\evnt{u}{c\,=\,z}}
              edge [color=blue,thick,->] (dx)
            }
            edge from parent [->,color=black] node [left] (bx) {\evnt{t}{b\,=\,x}}
          }
          edge from parent [->,color=black] node [left] {\evnt{t}{a\,=\,y}}
          edge [color=blue,thick,->] (bx)
        }
        edge from parent [->,color=black] node [left] {\evnt{s}{x\,=\,1}}
        edge [color=red,thick,<->,bend right=70] (bx)
        edge [color=red,thick,<->,bend right=55] (dx);
      }
      child { node[st] at (1.5,0) {}
        child { node[st] {}
          child { node[st,color=green] {}
            child { node[st,color=black] {}
              child { node[st,color=black,label=south:$E_2$] {}
                edge from parent [->,color=black] node [left] {\evnt{s}{x\,=\,1}}
              }
              edge from parent [->,color=black] node [left] {\evnt{u}{d\,=\,x}}
            }
            child { node[wu] {}
              child { node[fut,label=south:$E_3$] {}
                edge from parent [->,color=darkgray] node [left] {\evnt{u}{d\,=\,x}}
              }
              edge from parent [->,color=blue] node [right] {\evnt{s}{x\,=\,1}}
            }
            edge from parent [->,color=green] node [left] {\evnt{t}{b\,=\,x}}
          }
          child { node[wu] at (0.5,0) {}
            child { node[fut,color=blue] {}
              child { node[fut,label=south:$E_4$] {}
                edge from parent [->,color=darkgray] node [right] {\evnt{t}{b\,=\,x}}
              }
              edge from parent [->,color=blue] node [right] {\evnt{s}{x\,=\,1}}
            }
            edge from parent [->,color=blue] node [right] {\evnt{u}{d\,=\,x}}
          }
          edge from parent [->,color=green] node [left] {\evnt{u}{c\,=\,z}}
        }
        edge from parent [->,color=green] node [right] {\evnt{t}{a\,=\,y}}
      };
      \tikzwrapfigbg
    \end{tikzpicture}
  \end{minipage}
\caption{A program and its execution tree with
  the four executions that \OptimalDPOR will explore.
  In $E_1$, the red arcs show the conflict order; the blue arrows the program order. The first wakeup sequence is shown in green; the remaining two continue with~blue.}
\label{fig:rw}
\end{figure}
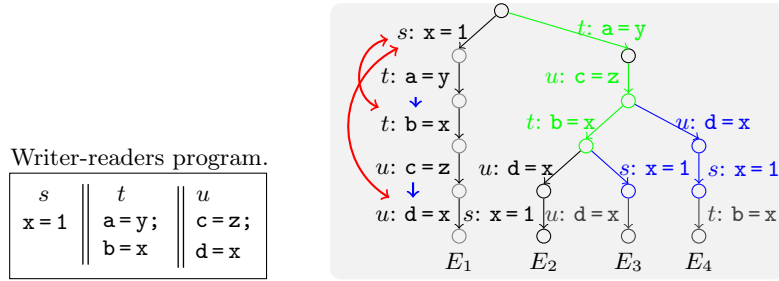

\subsection{Review of \OptimalDPOR}
Our DPOR algorithm for event-driven programs is an extension of \OptimalDPOR~\cite{abdulla2014optimal}.
Let us illustrate \OptimalDPOR on the program snippet shown in~\cref{fig:rw}.
In this code, three threads $s$, $t$, and~$u$ access
three shared variables \texttt{x}, \texttt{y}, and \texttt{z},\footnote{Throughout
  this paper, we assume that threads are spawned by a \texttt{main} thread,
  and that all shared variables get initialized to $0$, also by the main thread.}
whereas \texttt{a}, \texttt{b}, \texttt{c}, and \texttt{d} are thread-local registers.
\OptimalDPOR first explores a maximal execution, which it inspects to detect races.
From each race, it constructs an initial fragment of
an alternative execution which reverses the race and branches off from the explored execution just before the race.
Let us illustrate with the program in~\cref{fig:rw}.
Assume that the first execution is $\exseq_1$ (cf. the tree in~\cref{fig:rw}).
The DPOR algorithm first computes its happens-before order, denoted $\happbf{\hb}{E_1}$,
which is the transitive closure of the union of:
\begin{inparaenum}[(i)]
\item
  the \emph{program order}, which totally orders the events in each thread (small blue arrows to the left of $E_1$), and
\item
  the \emph{conflict order} which orders conflicting events: two events are conflicting if they access a common shared variable and at least one is a write
  (red arcs left of $E_1$).
\end{inparaenum}
A \emph{race} consists of two conflicting events in different threads that are adjacent in the $\happbf{\hb}{E_1}$-order.
The execution $\exseq_1$ contains two races (red arcs in~\cref{fig:rw}).
Let us consider the first race, in which the first event is \evnt{s}{x=1} and the second event is \evnt{t}{b=x}.
The alternative execution is generated by concatenating 
the sequence of events in $E_1$ that do not succeed the first event in the $\happbf{\hb}{E_1}$ order (i.e., $\evnt{t}{a\,=\,y}; \evnt{u}{c\,=\,z}$) with
the second event of the race \evnt{t}{b=x}.
This forms a \emph{wakeup sequence}, which branches off from $E_1$ just before
the race, i.e., at the beginning of the exploration (green in~\cref{fig:rw}).
The second race, between \evnt{s}{x=1} and~\evnt{u}{d=x} induces the wakeup sequence $t.u.u$ formed from the
sequence $\evnt{t}{a\,=\,y}; \evnt{u}{c\,=\,z}$ and the second event \evnt{u}{d\,=\,x}, also branching off at the beginning
(note that $t.u.u$ does not contain the second event \evnt{t}{b=x} of~$t$ since it succeeds \evnt{s}{x=1} in the $\happbf{\hb}{E_1}$-ordering).
When attempting to insert $t.u.u$, the algorithm will discover that this sequence is \emph{redundant}, since its events are
consistently contained in a continuation ($t.u.t.u$) of 
the already inserted wakeup sequence $t.u.t$, and it will therefore not insert $t.u.u$.
After this, the algorithm will reclaim the space for~$E_1$, extend $t.u.t$ into a maximal execution~$E_2$, 
in which races are detected that generate
two new wakeup sequences (which start in green and continue in blue), which are extended to two additional executions (cf.~\cref{fig:rw}).
\begingroup

\subsection{Challenges for Event-driven Programs}
A na\"ive way in which existing DPOR algorithms can handle event-driven programs is to consider all pairs of messages as conflicting.
However, such an approach is \emph{not} effective, since it will lead to exploration of all
different serialization orders of the messages, even if they are non-conflicting, as is the case for the
top left program of \cref{fig:example1new} in which two threads $s$ and $t$
post two messages $p_1$ and $p_2$ to a  handler thread~$h$.
(We show messages labeled by the message identifier and wrapped in brackets.)
Since the events of $p_1$ and $p_2$ are non-conflicting, 
exploring only one execution suffices.
In general, some messages of a program may be conflicting and some may not be, so a DPOR algorithm for event-driven programs
should explore only the necessary inequivalent orderings between conflicting messages.
\EventDPOR achieves this by extending \OptimalDPOR's technique for reversing races between events in different threads to
a mechanism for reversing races between events in different~messages.

\begin{figure}[t]
  \centering
  \begin{minipage}[b]{0.5\textwidth}
    \centering \footnotesize
    \begin{tikzpicture}[line width=1pt,framed,inner sep=1pt]
      \node[name=p,anchor=south west] at (-0.15,0.25) {{$s$}};
      \node[name=post1] at (0,0) {$\mathtt{post}(p_1,h)$};
      
      \draw[line width=0.5pt] ($(post1.north east)+(1pt,10pt)$)--($(post1.south east)+(1pt,-12pt)$);
      \draw[line width=0.5pt] ($(post1.north east)+(3pt,10pt)$)--($(post1.south east)+(3pt,-12pt)$);
    
      \node[name=q,anchor=south west] at (1.45,0.25) {{$t$}};
      \node[name=post2,anchor=west] at ($(post1.east)+(5pt,0.5pt)$) {$\mathtt{post}(p_2,h)$};
      
      \draw[line width=0.5pt] ($(post2.north east)+(1pt,10pt)$) -- ($(post2.south east)+(1pt,-12pt)$);
      \draw[line width=0.5pt] ($(post2.north east)+(3pt,10pt)$) -- ($(post2.south east)+(3pt,-12pt)$);
    
      \node[name=r,anchor=south west] at (2.9,0.25) {{$h$'s messages}};
      \node[name=m1,anchor=west] at ($(post2.east)+(8pt,0.5pt)$) {$p_1$: $\left[\texttt{x\,=\,1}\right]$};
      \node[name=m2,anchor=north west] at ($(m1.south west)+(0pt,-1pt)$) {$p_2$: $\left[\texttt{y\,=\,2}\right]$};
    \end{tikzpicture}
    \\[1.4em]
    \begin{tikzpicture}[line width=1pt,framed,inner sep=1pt]
      \node[name=p,anchor=south west] at (-0.15,0.25) {{$s$}};
      \node[name=post1] at (0,0) {$\mathtt{post}(p_1,h)$};
      
      \draw[line width=0.5pt] ($(post1.north east)+(1pt,10pt)$)--($(post1.south east)+(1pt,-60pt)$);
      \draw[line width=0.5pt] ($(post1.north east)+(3pt,10pt)$)--($(post1.south east)+(3pt,-60pt)$);
      
      \node[name=q,anchor=south west] at (1.45,0.25) {{$t$}};
      \node[name=post2,anchor=west] at ($(post1.east)+(5pt,0.5pt)$) {$\mathtt{post}(p_2,h)$};

      \draw[line width=0.5pt] ($(post2.north east)+(1pt,10pt)$) -- ($(post2.south east)+(1pt,-60pt)$);
      \draw[line width=0.5pt] ($(post2.north east)+(3pt,10pt)$) -- ($(post2.south east)+(3pt,-60pt)$);

      \node[name=r,anchor=south west] at (2.9,0.2) {{$h$}'s messages};
      \node[name=m1,anchor=west] at ($(post2.east)+(8pt,-12pt)$) {$p_1$: $\left[\begin{array}{@{}l@{}}\texttt{u\,=\,1};\\ \texttt{x\,=\,1};\\ \texttt{y\,=\,1}\end{array} \right]$};
      \node[name=m2,anchor=north west] at ($(m1.south west)+(0pt,-1pt)$) {$p_2$: $\left[\begin{array}{@{}l@{}}\texttt{v\,=\,2};\\ \texttt{a\,=\,x};\\ \texttt{b\,=\,y}\end{array}\right]$};
     \end{tikzpicture}
  \end{minipage}
  \begin{minipage}[b]{0.42\textwidth}
    \centering \footnotesize
    \begin{tikzpicture}[
        yscale=.57,xscale=.7,
        level distance=1cm,
        st/.style={ellipse,inner sep=0.2em,color=black,draw},
        wu/.style={ellipse,inner sep=0.2em,color=black,draw},
        child anchor=north
      ]
      \node[st] {}
      child { node[st] {}
        child { node[st] {}
          child { node[st] {}
            child { node[st] {}
              child { node[st] {}
                child { node[st] {}
                  child { node[st] {}
                    child { node[st,label=south:$E_1$] {}
                      edge from parent [->] node [left] (y2) {\hndlr{h}{p_2}{b\,=\,y}}
                    }
                    edge from parent [->] node [left] (x2) {\hndlr{h}{p_2}{a\,=\,x}}
                    edge [color=blue,thick,->] (y2)
                  }
                  edge from parent [->] node [left] {\hndlr{h}{p_2}{v\,=\,2}}
                  edge [color=blue,thick,->] (x2)
                }
                edge from parent [->] node [left] (y1) {\hndlr{h}{p_1}{y\,=\,1}}
                edge [color=red,thick,<->,bend right=92] (y2);
              }
              edge from parent [->] node [left] (x1) {\hndlr{h}{p_1}{x\,=\,1}}
              edge [color=blue,thick,->] (y1)
              edge [color=red,thick,<->,bend right=92] (x2);
            }
            edge from parent [->] node [left] {\hndlr{h}{p_1}{u\,=\,1}}
            edge [color=blue,thick,->] (x1)
          }
          child { node[wu, color=blue] {}
            child { node[st] {}
              child { node[st] {}
                child { node[st] {}
                  child { node[st] {}
                    child { node[st,label=south:$E_2$] {}
                      edge from parent [->] node [right] {\hndlr{h}{p_1}{y\,=\,1}}
                    }
                    edge from parent [->] node [right] {\hndlr{h}{p_1}{x\,=\,1}}
                  }
                  edge from parent [->] node [right] {\hndlr{h}{p_1}{u\,=\,1}}
                }
                edge from parent [->, color=black] node [right] {\hndlr{h}{p_2}{b\,=\,y}}
              }
              edge from parent [->,color=blue] node [right] {\hndlr{h}{p_2}{a\,=\,x}}
            }
            edge from parent [->,color=blue] node [right] {\hndlr{h}{p_2}{v\,=\,2}}
          }
          edge from parent [->] node [left] {\evnt{t}{post($p_2$,$h$)}}
        }
        edge from parent [->] node [left] {\evnt{s}{post($p_1$,$h$)}}
      };
      \tikzwrapfigbg
    \end{tikzpicture}
  \end{minipage}
  \caption{An event-driven program with non-conflicting messages (top left). A program with non-atomic conflicting messages (bottom left)
    and its tree of executions (right).}
  \label{fig:example1new}
\end{figure}
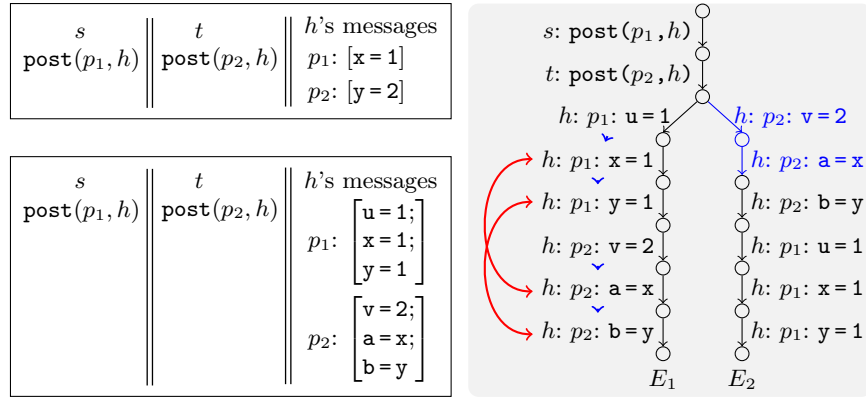

We illustrate this mechanism on the program at the bottom left of \cref{fig:example1new}.
Assume that the first explored execution is $E_1$. It contains two races between events in the two messages, one on \texttt{x} and one on \texttt{y}.
According to \OptimalDPOR's principle for race reversal, the race on \texttt{x} should induce an alternative execution composed of 
the sequence of events that do not happen-after the first event (i.e., {\hndlr{h}{p_1}{u\,=\,1}} {\hndlr{h}{p_2}{v\,=\,2}}) and the 
second event {\hndlr{h}{p_2}{a\,=\,x}}
(for brevity, we do not show the two post events).
However, since message execution is serialized, these events cannot form an execution.
Therefore, \EventDPOR forms the alternative execution (shown in blue) by appending the second event {\hndlr{h}{p_2}{a\,=\,x}} to a 
maximal subset of the events of $E_1$ which is closed under $\happbf{\hb}{E_1}$-predecessors
(i.e., if it contains an event $e$ then it also contains all its $\happbf{\hb}{E_1}$-predecessors), and which can form an execution that does not
contain the first event. 
Later, this wakeup sequence is extended to execution $E_2$.
Let us then consider the race on \texttt{y}.
The constructed wakeup sequence should append the second event {\hndlr{h}{p_2}{b\,=\,y}} to a maximal subset of events
that do not happen-after the first event {\hndlr{h}{p_1}{y\,=\,1}}.
However, there is no execution that satisfies these constraints, since it would have to include
{\hndlr{h}{p_2}{a\,=\,x}} before its $\happbf{\hb}{E_1}$-predecessor {\hndlr{h}{p_1}{x\,=\,1}}.
The conclusion is that the race on \texttt{y} cannot (and should not) be considered for reversal, whereas that on \texttt{x} should be reversed.
More generally, if two messages executing on the same handler thread are in conflict, then a wakeup sequence is constructed consisting of only the second message up until and including its first conflicting event.

When messages can branch on values read from shared variables,
reversing the order of two messages may change the control flow of each involved message.
Also in this case, \EventDPOR's principles for reversing races work fine.
We illustrate this on the program in \cref{fig:non-atomic-1},
consisting of two threads $s$ and $t$ and a handler thread $h$.
Thread $s$ posts a message $p_1$ to $h$ and thereafter writes to \texttt{x}.
Thread $t$ posts message $p_2$ to $h$
that reads from \texttt{x} and \emph{may} then read from~\texttt{y}.
 
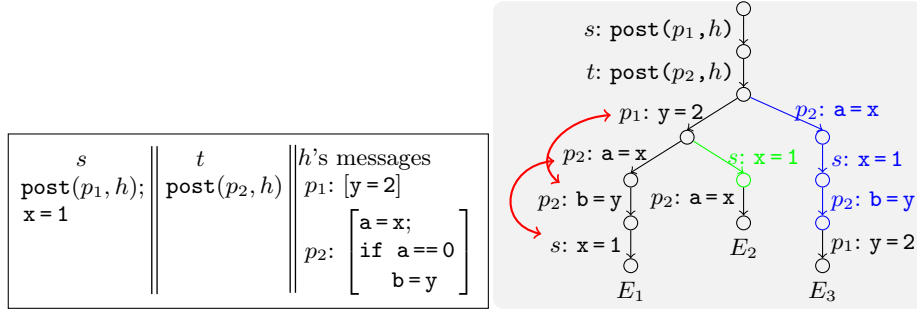
\begin{figure}[t]
  \centering
  \begin{minipage}[b]{0.52\textwidth}
    \centering \footnotesize
    \begin{tikzpicture}[line width=1pt,framed,inner sep=1pt]
      \node[name=p,anchor=south west] at (-0.15,0.25) {{$s$}};
      \node[name=post1] at (0,0) {$\mathtt{post}(p_1,h);$};
      \node[name=n12,anchor=north] at ($(post1.south)+(-14pt,0.5pt)$) {$\texttt{x\,=\,1}$};
      
      \draw[line width=0.5pt] ($(post1.north east)+(1pt,10pt)$)--($(post1.south east)+(1pt,-34pt)$);
      \draw[line width=0.5pt] ($(post1.north east)+(3pt,10pt)$)--($(post1.south east)+(3pt,-34pt)$);
      
      \node[name=q,anchor=south west] at (1.45,0.25) {{$t$}};
      \node[name=post2,anchor=west] at ($(post1.east)+(5pt,0.5pt)$) {$\mathtt{post}(p_2,h)$};

      \draw[line width=0.5pt] ($(post2.north east)+(1pt,10pt)$) -- ($(post2.south east)+(1pt,-34pt)$);
      \draw[line width=0.5pt] ($(post2.north east)+(3pt,10pt)$) -- ($(post2.south east)+(3pt,-34pt)$);
      
      \node[name=r,anchor=south west] at (2.8,0.2) {{$h$}'s messages};
      \node[name=m1,anchor=west] at ($(post2.east)+(5pt,0.5pt)$)  {$p_1$: $\left[\texttt{y\,=\,2}\right]$};
      \node[name=m2,anchor=north west] at ($(m1.south west)+(0pt,-1pt)$) {$p_2$: $\left[\begin{array}{@{}l@{}}\texttt{a\,=\,x};\\ \texttt{if\ a\,==\,0\,}\\ \ \ \ \ \texttt{b\,=\,y}\end{array}\right]$};
    \end{tikzpicture}
  \end{minipage}
  \begin{minipage}[b]{0.46\textwidth}
    \centering \footnotesize
    \begin{tikzpicture}[
        yscale=.57,xscale=1.0,
        level distance=1cm,
        st/.style={ellipse,inner sep=0.2em,draw},
        wu/.style={ellipse,inner sep=0.2em,color=blue,draw},
        wug/.style={ellipse,inner sep=0.2em,color=green,draw},
        child anchor=north
      ]
      \node[st] {}
      child { node[st] {}
        child { node[st] {}
          child { node[st] {}
            child { node[st] {}
              child { node[st] {}
                child { node[st,label=south:$E_1$] {}
                  edge from parent [->] node [left] (x1) {\evnt{s}{x\,=\,1}}
                }
                edge from parent [->] node [left] (by) {\evnt{p_2}{b\,=\,y}}
              }
              edge from parent [->] node [left] {\evnt{p_2}{a\,=\,x}}
              edge [color=red,thick,<->,bend right=75] (x1);
            }
            child { node[wug] {}
              child { node[st,color=black,label=south:$E_2$] {}
                edge from parent [->,color=black] node [left] {\evnt{p_2}{a\,=\,x}}
              }
              edge from parent [->,color=green] node [right] {\evnt{s}{x\,=\,1}}
             }
            edge from parent [->] node [left] {\evnt{p_1}{y\,=\,2}}
            edge [color=red,thick,<->,bend right=60] (by);
          }
          child { node[wu] at (0.3,0) {}
            child { node[wu] {}
              child { node[st] {}
                child { node[st,color=black,label=south:$E_3$] {}
                  edge from parent [->,color=black] node [right] {\evnt{p_1}{y\,=\,2}}  
                }
                edge from parent [->,color=blue] node [right] {\evnt{p_2}{b\,=\,y}}  
              }
              edge from parent [->,color=blue] node [right] {\evnt{s}{x\,=\,1}}
            }
            edge from parent [->,color=blue] node [right] {\evnt{p_2}{a\,=\,x}}
          }
          edge from parent [->] node [left] {\evnt{t}{post($p_2$,$h$)}}
        }
        edge from parent [->] node [left] {\evnt{s}{post($p_1$,$h$)}}
      };
      \tikzwrapfigbg
    \end{tikzpicture}
  \end{minipage}
  \caption{A program with messages that branch on read values and its exploration tree.}
  \label{fig:non-atomic-1} 
\end{figure}

Assume that the first execution is $E_1$, where $s$'s access to \texttt{x} goes last.
The execution has two races:
one on \texttt{y} between {\evnt{p_1}{y\,=\,2}} and {\evnt{p_2}{b\,=\,y}}, and
one on \texttt{x} between {\evnt{p_2}{a\,=\,x}} and \evnt{s}{x\,=\,1}.
The race on \texttt{x} can be handled in the same way as in \OptimalDPOR:
the wakeup sequence is $\evnt{s}{x\,=\,1}$, which branches off after the prefix $s.t.p_1$ (green in~\cref{fig:non-atomic-1}), and 
will subsequently be extended to execution~$E_2$.
The race on \texttt{y} is a race between events in two messages on the same handler thread. As in the previous example, the wakeup sequence will
include the second message up until and including the first racing event, which is {\evnt{p_2}{b\,=\,y}}.
Included in the events that do not happen-after the first event is also $\evnt{s}{x\,=\,1}$, which must be placed after its predecessor {\evnt{p_2}{a\,=\,x}}, yielding the wakeup sequence
  {\evnt{p_2}{a\,=\,x}}; \evnt{s}{x\,=\,1}; {\evnt{p_2}{b\,=\,y}}, which
branches off after
  \evnt{s}{post($p_1$,$h$)}, 
  \evnt{t}{post($p_2$,$h$)}.
  This is the blue rightmost branch of the tree in~\cref{fig:non-atomic-1},
  and is later extended into the execution $E_3$.
Execution $E_3$ has a race on \texttt{x}. Its reversal produces the wakeup sequence \evnt{s}{x\,=\,1}, which is a tentative branch next
to  {\evnt{p_2}{a\,=\,x}}. However, this wakeup sequence is not in conflict with the left branch labeled  {\evnt{p_1}{b\,=\,y}}, which means that it will not
be inserted for the reason that it is equivalent to a subsequence of an execution starting with  {\evnt{p_1}{b\,=\,y}}, namely $E_2$.

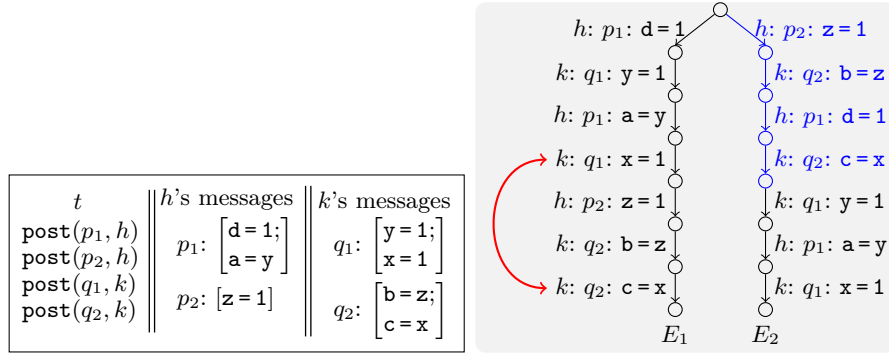
\begin{figure}[t]
  \centering
  \begin{minipage}[b]{0.5\textwidth}
    \footnotesize
    \begin{tikzpicture}[line width=1pt,framed,inner sep=1pt]
      \node[name=p,anchor=south west] at (-0.15,0.25) {{$t$}};
      \node[name=postp1] at (0,0) {$\mathtt{post}(p_1,h)$};
      \node[name=postp2,anchor=west] at ($(postp1.south west)+(0pt,-4pt)$) {$\mathtt{post}(p_2,h)$};
      \node[name=postq1,anchor=west] at ($(postp2.south west)+(0pt,-4pt)$) {$\mathtt{post}(q_1,k)$};
      \node[name=postq2,anchor=west] at ($(postq1.south west)+(0pt,-4pt)$) {$\mathtt{post}(q_2,k)$};
      
      \draw[line width=0.5pt] ($(post1.north east)+(1pt,10pt)$)--($(post1.south east)+(1pt,-32pt)$);
      \draw[line width=0.5pt] ($(post1.north east)+(3pt,10pt)$)--($(post1.south east)+(3pt,-32pt)$);
      
      \node[name=q,anchor=south west] at (1.0,0.25) {{$h$}'s messages};
      \node[name=p1,anchor=west] at ($(post1.east)+(10pt,-5pt)$) {$p_1$:  $\left[\begin{array}{@{}l@{}}\texttt{d\,=\,1};\\ \texttt{a\,=\,y}\end{array} \right]$};
      \node[name=p2,anchor=west] at ($(p1.south west)+(0pt,-8pt)$) {$p_2$: $\left[\texttt{z\,=\,1}\right]$};

      \draw[line width=0.5pt] ($(p1.north east)+(3pt,10pt)$) -- ($(p1.south east)+(3pt,-22pt)$);
      \draw[line width=0.5pt] ($(p1.north east)+(5pt,10pt)$) -- ($(p1.south east)+(5pt,-22pt)$);

      \node[name=r,anchor=south west] at (3.1,0.2) {{$k$}'s messages};
      \node[name=q1,anchor=west] at ($(p1.east)+(12pt,0pt)$) {$q_1$: $\left[\begin{array}{@{}l@{}}\texttt{y\,=\,1};\\ \texttt{x\,=\,1}\end{array} \right]$};
      \node[name=q2,anchor=west] at ($(q1.south west)+(0pt,-12pt)$) {$q_2$: $\left[\begin{array}{@{}l@{}}\texttt{b\,=\,z};\\ \texttt{c\,=\,x}\end{array} \right]$};
    \end{tikzpicture}
  \end{minipage}
  \begin{minipage}[b]{0.45\textwidth}
    \centering 
    \footnotesize
    \begin{tikzpicture}[
        yscale=.57,xscale=0.8,
        level distance=1cm,
        st/.style={ellipse,inner sep=0.2em,draw},
        wu/.style={ellipse,inner sep=0.2em,color=blue,draw},
        child anchor=north
      ]
      \node[st] {}
      child { node[st] {}
        child { node[st] {}
          child { node[st] {}
            child { node[st] {}
              child { node[st] {}
                child { node[st] {}
                  child { node[st,label=south:$E_1$] {}
                    edge from parent [->] node [left] (cx)  {\hndlr{k}{q_2}{c\,=\,x}}
                  }
                  edge from parent [->] node [left] {\hndlr{k}{q_2}{b\,=\,z}}
                }
                edge from parent [->] node [left] (y1)  {\hndlr{h}{p_2}{z\,=\,1}}
              }
              edge from parent [->] node [left] {\hndlr{k}{q_1}{x\,=\,1}}
              edge [color=red,thick,<->,bend right=90] (cx);
            }
            edge from parent [->] node [left] (ay) {\hndlr{h}{p_1}{a\,=\,y}}
          }
          edge from parent [->] node [left] {\hndlr{k}{q_1}{y\,=\,1}}
        }
        edge from parent [->] node [left] {\hndlr{h}{p_1}{d\,=\,1}}
      }
      child { node[wu] {}
        child { node[wu] {}
          child { node[wu] {}
            child { node[wu] {}
              child { node[st, color=black] {}
                child { node[st] {}
                  child { node[st,label=south:$E_2$] {}
                    edge from parent [->] node [right] {\hndlr{k}{q_1}{x\,=\,1}}
                  }
                  edge from parent [->] node [right] {\hndlr{h}{p_1}{a\,=\,y}}
                }
                edge from parent [->,color=black] node [right] (y1) {\hndlr{k}{q_1}{y\,=\,1}}
              }
              edge from parent [->,color=blue] node [right] {\hndlr{k}{q_2}{c\,=\,x}} 
            }
            edge from parent [->,color=blue] node [right] (ay)  {\hndlr{h}{p_1}{d\,=\,1}} 
          }
          edge from parent [->,color=blue] node [right] {\hndlr{k}{q_2}{b\,=\,z}} 
        }
        edge from parent [->,color=blue] node [right] {\hndlr{h}{p_2}{z\,=\,1}} 
      };
      \tikzwrapfigbg
    \end{tikzpicture}
  \end{minipage}
  \caption{A program in which a reversal of the race on \texttt{x} will reorder messages on the handler $k$,
    and two executions that will be explored.}
  \label{fig:example3}
\end{figure}
\paragraph{Reordering Messages when Reversing Races}
\EventDPOR's principles for reversing races may necessitate
reordering of messages on handlers that are not involved in the race.
Consider the program in~\cref{fig:example3}.
Assume that the first explored execution is~$E_1$, where we have omitted the initial sequence of post events of thread $t$ for  succinctness.
In $E_1$, message $p_1$ is processed before $p_2$, and $q_1$ is processed before $q_2$. There are three races in $E_1$, one on each of the shared variables \texttt{x}, \texttt{y}, \texttt{z}.
Let us consider the race on \texttt{x}, shown by the red arrow. A wakeup sequence which reverses  this race must include all events of $q_2$, since these are the
$\happbf{\hb}{E_1}$-predecessors of \evnt{q_2}{c\,=\,x}. It must also include
the write to \texttt{z} by $p_2$ since it is a $\happbf{\hb}{E_1}$-predecessor of events in $q_2$. On the other hand, it cannot include any part of the message $q_1$, since $q_1$ must now occur after $q_2$, and therefore it also cannot include the read of \texttt{y} by $p_1$ since its predecessor in $q_1$ is missing.
In summary, the wakeup sequence contains two fully processed messages $p_2$ and~$q_2$, the event \hndlr{h}{p_1}{d\,=\,1} of $p_1$, 
but no events from $q_1$. Such a wakeup sequence must branch off after the post events of $t$, i.e., from the root of the tree to the right in~\cref{fig:example3}. Later, this wakeup sequence is extended to a full execution $E_2$.
In total, the program of~\cref{fig:example3} has eight inequivalent executions (the other six are not shown).

\section{Computation Model}
\label{sec:model}

\hasbeenremoved{In this section, we introduce the class of event-driven programs that we consider.
We also define several semantical notions such as transitions, events, executions, and traces.}

\subsection{Programs}
\label{sec:programs}
We consider programs consisting of a finite set of \emph{threads} that interact via
a finite set of \emph{(shared) variables}.
Each thread is either a \emph{normal thread} or a \emph{handler thread}.
%
A normal thread has a finite set of local registers and runs a deterministic code, built in a standard way
from expressions and atomic statements, using standard control
flow constructs (sequential composition, selection and bounded iteration).
Atomic statements read or write to shared variables and local registers,
including read-modify-write operations, such as \mbox{compare-and-swap}.
A handler thread has a \emph{mailbox} to which all threads (also handler threads) can post messages.
A mailbox has unbounded capacity, implying that the posting of a message to a mailbox can never block.
A message consists of a deterministic code, built in the same way as the code of a thread. 
We let $\post(p,h)$ denote the statement which posts the message $p$ into the mailbox of handler thread $h$.
A handler thread repeatedly extracts
a message from its mailbox, executes the code of the message to completion,
then extracts a next message and executes its code, and so on. 
Messages are extracted from the mailbox in arbitrary order.
The execution of a message is interleaved with the statements of other threads.

The local state of a thread is a valuation of its local registers
together with the contents of its mailbox.
A global state of a program consists of a local state of each thread together
with a valuation of the shared variables.
The program has a unique initial state, in which mailboxes are empty.

Recall that we use \emph{message} to denote what is called \emph{event} in \cref{sec:intro}.

\subsection{Events, Executions, Happens-before Ordering, and Equivalence}
\label{sec:events}

We use $s,t, \ldots$ for threads,
$p,q,\ldots$ for messages and non-handler threads,
\texttt{x}, \texttt{y}, \texttt{z} for shared variables,
and \texttt{a}, \texttt{b}, \texttt{c}, \texttt{d} for local registers.
We assume, 
wlog,
that the first event of a message does not access a shared variable, but only performs a local action, e.g., related to initialization
of message execution.
In order to simplify the presentation, we henceforth extend the term \emph{message} to refer not only to a message but also to a non-handler thread.

The execution of a program statement is an \emph{event},
which affects the global state of the program. 
An event  is denoted by a pair $\transpair pi$, where $p$ denotes the
message containing the event and $i$ is a positive integer,
denoting that the event results from the $i$-th execution step in message $p$.
\hasbeenremoved{
We call an event \emph{global} if it accesses a shared variable or posts a message.
A message is denoted by the event that posted it. 
For instance, if the second step of
message $p$ posts a message, whose first step is to post another message $q$,
then the event representing the first step of message $q$ is denoted
$\transpair{\msgof{\msgof{p}{2}}{1}}{1}$.
As is customary in DPOR algorithms, we can let an event
represent the combined effect of a sequence of statements, if at most one of them affects a shared variable or the
local state of other threads.
This avoids consideration of interleavings of
local statements of different threads in the analysis.
}
An \emph{execution sequence} $\exseq$
is a finite sequence of events, starting from the initial state of the program.
Since thread and message codes are deterministic, an execution
sequence $\exseq$ can be uniquely characterized by the sequence
of messages (and non-handler threads) that perform execution steps in $\exseq$,
where we use dot(.) as concatenation operator.
Thus $p.p.q$ denotes the
execution sequence consisting first of two events of $p$, followed by an event of~$q$. 

We let $\enabled{E}$ denote the set of
messages that can perform a next event in the state to which $E$ leads.
A sequence $\exseq$ is \emph{maximal} if $\enabled{E} = \emptyset$.
We use $u,v,w, \ldots$ to range over sequences of events. 
We introduce the following notation, where $E$ is an execution sequence and $w$ is a sequence of events.
\begin{itemize}[-]
\item $\emptyseq$  denotes the empty sequence. 
\item $\valid{\exseq}{w}$ denotes that $\exseq.w$ is an execution sequence.
\item $w \remove p$ denotes the sequence
  $w$ with its first occurrence of $p$ (if any) removed.
\item $\dom{\exseq}$ denotes the set of events $\transpair{p}{i}$ in $\exseq$, that is, $\transpair{p}{i} \in \dom{\exseq}$ iff
  $\exseq$ contains at least $i$ events of $p$.
  We also write $e \in E$ to denote $e \in \dom{E}$.
\item $\nextev{\exseq}{p}$ denotes the next event to be performed by the message $p$ after the execution $E$ if $p \in \enabled{E}$,
  otherwise $\nextev{\exseq}{p}$ is undefined.
\item $\procof{\event}$ denotes the
  message that performs $e$, i.e., $e$ is of form
$\event = \transpair{\procof{\event}}{i}$ for some $i$.
\item $\exseq' \prefix \exseq$ denotes that  $\exseq'$ is 
  a  (not necessarily strict) prefix of $\exseq$.
\end{itemize}
We say that \emph{$p$ starts after $E$} if $p$ has been posted in $E$, but not yet performed any events in $E$.
We say that \emph{$p$ is active after $E$} if $p$ has been posted in $E$, but not finished its execution in $E$.

\hasbeenremoved{The basis for our DPOR algorithm is the definition of a happens-before relation on the events of each execution sequence,
which captures the data and control dependencies that must be respected by any equivalent execution.}

\begin{definition}[Happens-before]
\label{def:hb-def}
Given an execution sequence $E$,
we define the \emph{happens-before relation} on $E$, denoted
$\happbf{\hb}{E}$, as the smallest irreflexive partial order on $\dom{E}$ such that
$e \happbf{\hb}{E} e'$ if $\event$ occurs before $\event'$ in $\exseq$ and either  
\begin{itemize}
\item
$e$ and $e'$ are performed by the same message $p$,
\item
$e$ and $e'$ access a common shared variable \texttt{x} and at least one writes to \texttt{x}, or
\item
$\procof{e'}$ is the message that is posted by $e$ and $e'$ is the first event of $\procof{e'}$. 
  \qed
\end{itemize}
\end{definition}
\hasbeenremoved{Intuitively, $\happbf{\po}{E}$ (\emph{program order}) is the total order of events of each message.
Note that $\happbf{\po}{E}$  does not order events of different messages relative to each other.
The relation $\happbf{\cnf}{E}$ (\emph{conflicts with}) captures data flow constraints arising from reads and writes to shared variables.
The relation $\happbf{\pb}{E}$ (\emph{posted by}) captures the causal dependency from message posting to message execution.}
The \emph{\hb-trace} (or \emph{trace} for short) of $E$ is the directed graph $(\dom{E},\ \happbf{\hb}{E})$.
\begin{definition}[Equivalence]
\label{def:hb-equiv}
Two execution sequences $E$ and $E'$ are 
\emph{equivalent}, denoted $E \mtequiv E'$, if they have the same trace.
We let \eqclass{E} denote the equivalence class of $E$.
  \qed
\end{definition}
Note that for programs that do not post or process messages,
$\mtequiv$ is the standard Mazurkiewicz trace equivalence for multi-threaded programs \cite{Mazurkiewicz:traces,Godefroid:thesis,FG:dpor,abdulla2014optimal}.
We say that two sequences of events, $w$ and $w'$, with $\valid{E}{w}$ and $\valid{E}{w'}$, are
\emph{equivalent after $E$}, denoted $w \equivafter{E} w'$ if $E.w \mtequiv E.w'$.

\section{The \EventDPOR Algorithm}
\label{sec:eventdpor}

In this section, we present \emph{\EventDPOR}, a DPOR algorithm for event-driven programs.
Given a terminating program on given input,
the algorithm explores different maximal executions resulting from different thread interleavings.

\subsection{Central Concepts in \EventDPOR}
\label{sec:prels}
\hasbeenremoved{In this section, we define central concepts in \EventDPOR. We first define
the concepts of \emph{happens-before prefix} and \emph{weak initials},
which are used in the check for redundancy of new executions.
Thereafter, we define \emph{races}: these are used to construct new executions from already explored ones.}

\begin{definition}[Happens-before Prefix]
\label{def:hb-prefix}
Let $E$ and $E'$ be execution sequences.
We say that $E'$ is a {\em happens-before prefix} of $E$, denoted
$E' \mtprefix E$, if
\begin{inparaenum}[(i)]
\item
  $\dom{E'} \subseteq \dom{E}$,
\item
  $\happbf{\hb}{E'}$ is the restriction of $\happbf{\hb}{E}$ to $E'$, and
\item
  whenever $e \happbf{\hb}{E} e'$ for some $e' \in \dom{E'}$, then $e \in \dom{E'}$.
\end{inparaenum}
We let $w' \mtprefixafter{E} w$ denote that $E.w' \mtprefix E.w$.
  \qed
\end{definition}
Intuitively, $E' \mtprefix E$ denotes that the execution $E'$ is ``contained'' in the execution $E$
in such a way that it is not affected by the events in $E$ that are not in $E'$.
\footnote{The  relation $w' \mtprefixafter{E} w$ is also introduced in \citet{Maiya:tacas16}, as ``$w$ is a dependence-covering sequence of $w'$.''}
To illustrate, for the top left program of \cref{fig:example1new},
the execution $E'$ consisting of \evnt{t}{post($p_2$,$h$)} \hndlr{h}{p_2}{y\,=\,2} is a happens-before prefix of
any maximal execution of the program, since the event of $p_2$ cannot happen-after any other event than
the event that posts $p_2$, which is already in $E'$.

\begin{definition}[Weak Initials]
  \label{def:winits}
  Let $E$ be an execution sequence, and $w$ be a sequence with $\valid Ew$.
  The set $\winits{\exseq}{w}$ of \emph{weak initials of $w$ after $E$} is the set of messages $p$
such that $\valid \exseq p.w'$ for some $w'$ with $w \mtprefixafter{E} p.w'$.
\qed
\end{definition}
\begingroup
\setlength{\intextsep}{0em}%
\setlength{\columnsep}{.75em}%
Intuitively, $p$ is in $\winits{\exseq}{w}$ if $p$ can execute the first event in a continuation of $\exseq$ which ``contains'' $w$, in the sense of $\sqsubseteq$.
In \EventDPOR, the concept of weak initials is used to test whether a new sequence is redundant, i.e., is ``contained in'' an execution that have been explored or in
a wakeup sequence that is scheduled for exploration.
Note that in \cref{def:winits}, we can generally not choose $w'$ as $w \remove p$.
This happens, e.g., if $p$ does not occur in $w$ but instead $w$ contains another message $p'$ which executes on the same handler as $p$ and
does not conflict with $p$; in this case $w'$ must contain a completed execution of $p$ inserted before $p'$.

\begin{wrapfigure}{r}{0.46\textwidth}
  \footnotesize
  \begin{tikzpicture}[line width=1pt,framed,inner sep=1pt]
    \node[name=p,anchor=south west] at (-0.15,0.25) {{$s$}};
   \node[name=post1] at (0,0) {$\mathtt{post}(p_1,h)$};
    
    \draw[line width=0.5pt] ($(post1.north east)+(1pt,10pt)$)--($(post1.south east)+(1pt,-22pt)$);
    \draw[line width=0.5pt] ($(post1.north east)+(3pt,10pt)$)--($(post1.south east)+(3pt,-22pt)$);
    
    \node[name=q,anchor=south west] at (1.45,0.25) {{$t$}};
    \node[name=post2,anchor=west] at ($(post1.east)+(5pt,0.5pt)$) {$\mathtt{post}(p_2,h)$};
    
    \draw[line width=0.5pt] ($(post2.north east)+(1pt,10pt)$) -- ($(post2.south east)+(1pt,-22pt)$);
    \draw[line width=0.5pt] ($(post2.north east)+(3pt,10pt)$) -- ($(post2.south east)+(3pt,-22pt)$);
    
    \node[name=r,anchor=south west] at (2.7,0.25) {{$h$'s messages}};
    \node[name=m1,anchor=west] at ($(post2.east)+(8pt,0.5pt)$) {$p_1$: $\left[\texttt{x\,=\,1}\right]$};
    \node[name=m2,anchor=north west] at ($(m1.south west)+(0pt,-1pt)$) {$p_2$: $\left[\begin{array}{@{}l@{}}\texttt{y\,=\,2};\\ \texttt{z\,=\,2}\end{array}\right]$};
  \end{tikzpicture}
\vspace{-0.6cm}
  \caption{Illustrating weak initials}
  \label{prog:wi-illustration}
\end{wrapfigure}
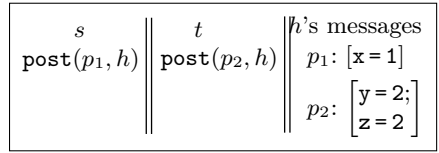
We illustrate using the program shown on the right.
If we let $E$ be the execution $s.t$ and $w$ be the sequence $p_1$,
we have $p_2 \in \winits{\exseq}{w}$, since $w \mtprefixafter{E} p_2.p_2.p_1$.
This illustration shows that in order to determine whether $p \in \winits{\exseq}{w}$ for a message $p$, one must know which shared-variable access will be performed by
$\nextev{E}{p}$, and, in case $p$ starts after $E$ but will execute after some other message on its handler,
also the sequences of shared-variable accesses that $p$ will perform when executing to completion.

The weak initial check problem consists in checking whether $p \in \winits{\exseq}{w}$.

\begin{theorem}
\label{thm-lowerbound}
The weak initial check problem is NP-hard. 
\end{theorem}

The proof of the above theorem can be found in \cref{thm-lowerbound-weak}.
In \cref{sec:checkwi}, we propose a sequence of inexpensive rendundancy checks, which have shown to be sufficient for all our benchmarks.

\endgroup 

\begin{definition}[Races]
\label{def:races}
Let $E$ be a maximal execution sequence.
Two events $e$ and $e'$ in different messages are in a \emph{race}, denoted $e\revrace{E}e'$, if $e \happbf{\hb}{E} e'$ and
\begin{enumerate}[(i)]
\item
$e$ and $e'$ access a common shared variable and at least one is a write, and
\item  there is no event $e''$ with $e \happbf{\hb}{E} e''$ and $e'' \happbf{\hb}{E} e'$.
  \qed
\end{enumerate}
\end{definition}
Intuitively, a race arises between conflicting accesses to a shared variable, by events which are in different messages but adjacent
in the $\happbf{\hb}{E}$ order.

\subsection{The \EventDPOR Algorithm}
\label{sec:algo:access-sets}

The \EventDPOR algorithm, shown as pseudocode in \cref{alg:eventdpor-access},
performs a depth-first exploration of executions using the recursive procedure
$\explore(\exseq)$, where $\exseq$ is the currently explored execution,
which also serves as the stack of the exploration.
In addition the algorithm maintains three mappings from prefixes of $\exseq$, named $\done$, $wut$, and $\pendingwusname$.
For each prefix $\exseq'$ of $\exseq$, 
\begin{itemize}
\item $\done(\exseq')$ is a mapping whose domain is the set of messages $p$ for which the call $\explore(\exseq'.p)$ has returned.
If $p$ does not start after $E'$, then $\done(\exseq')(p)$ is the shared variable-access performed by $\nextev{E'}{p}$.
If $p$ starts after $E'$, then $\done(\exseq')(p)$ is the set of sequences of shared variable-accesses that can be performed in a completed
    execution of $p$ after $E'$.
   The information in $\done(\exseq')(p)$ is collected during the call $\explore(\exseq'.p)$
  (\crefrange{algacsl:collection-start}{algacsl:donesleeptree-add-ev}).
  \item $\wut{\exseq'}$ is a \emph{wakeup tree}, i.e., an ordered tree $\tuple{B,\prec}$ where
    $B$ is a prefix-closed set of sequences, whose leaves are wakeup sequences.
    For each sequence $u \in B$, the order $\prec$ orders its children (of form  $u.p$)
    by the order in which they were added to $\wut{\exseq'}$. This is also the order in which the sequences of form
    $\exseq'.u.p$ will be visited in the recursive exploration.
\hasbeenremoved{We extend $\prec$ to the post-order relation on $B$ induced by the ordering $\prec$ on children of a node.}
  \item $\pendingwus{\exseq'}$ is a set of wakeup sequences $v$ that were previously being inserted into some wakeup tree $\wut{\exseq''}$, but
    were ``parked'' at the sequence $\exseq'$ because at that time there was not enough information to determine where in $\wut{\exseq''}$ to place $v$.
    Later, when a branch of $\wut{\exseq''}$ has been extended to a maximal execution,
    it should be possible to determine where to insert $v$.
\end{itemize}
\hasbeenremoved{The already explored executions together with the sequences in the wakeup trees can be thought of as forming an exploration tree $\exseqs$.}


Each call to $\explore(E)$ first initializes $\done(E)$ and $\pendingwus{E}$
($\wut{\exseq}$ was initialized before the call),
and thereafter enters one of two phases:
\emph{race detection} (\crefrange{algacsl:event-race-begin}{algacsl:event-race-end}) or
\emph{exploration} (\crefrange{algacsl:exploration-begin}{algacsl:donesleeptree-add-ev}).
The race detection phase is invoked when $\exseq$ is a maximal execution sequence.
First, for each wakeup sequence $v$ parked at a prefix $E'$ of $E$ it invokes $\insertpendingwu{v}{E'}$ to insert $v$ into the appropriate wakeup tree
(\crefrange{algacsl:insert-parkedwus-begin}{algacsl:insert-pendingwu}).
Thereafter, each race (of form $\event \revrace{\exseq} \event'$) in $\exseq$ is analyzed by $\reverserace(\exseq,e,e')$, which returns a set of
executions that reverse the race. Each such execution $E'.v$ is returned as a pair $\tuple{E',v}$, where $v$ is a wakeup sequence that should be
considered for insertion in the wakeup tree at $E'$.
Each wakeup sequence $v$ is checked for redundancy (\cref{algacsl:event-test}), using
the information in $\done$.
If $v$ is not redundant, it is inserted 
into the wakeup tree at $E'$ for future exploration (\cref{algacsl:event-insert}).
\hasbeenremoved{Wakeup tree insertion is elaborated below, and in~\cref{alg:wakeuptree}.}

\begin{algorithm}[!htp]
\Initial{$\explore(\emptyseq)$~with~$\wut{\emptyseq}=\emptytree$}
\BlankLine
\Fn{$\explore(\exseq)$\tcp*[f]{Returns access sequences of messages}}{
$\done(\exseq) := \emptyset$\;\label{algacsl:doneset-init}
$\pendingwus{\exseq} := \emptyset$\;\label{algacsl:pendingwus-init}
  \If(\tcp*[f]{When $E$ is maximal, enter race detection}){$\enabled{\exseq} = \emptyset$}{
    \label{algacsl:event-race-begin}
    \ForEach{$\exseq' \leq \exseq$}{\label{algacsl:insert-parkedwus-begin}
      \ForEach(\tcp*[f]{Parked wakeup sequences}){$v \in \pendingwus{\exseq'}$}{
        $\insertpendingwu{v}{E'}$\tcp*{are inserted at the appropriate place}\label{algacsl:insert-pendingwu}
      }
    }
    \ForEach(\tcp*[f]{For each race in $\exseq$}){$\event,\event'\ \keyword{such that} \ \event \revrace{\exseq} \event'$\label{algacsl:race-loop}} {
      \ForEach(\tcp*[f]{For each race reversal}){$\tuple{\exseq',v} \in \reverserace(\exseq,e,e')$\label{algacsl:rev-race}}{
        \If(\tcp*[f]{If $v$ is not redundant}){$\neg \exists E'',w,p \ \mbox{ s.t. }\ E''.w = E' \land p \in \dom{\done(E'')} \land p \in \winits{E''}{w.v}$\label{algacsl:event-test}}{
          $\insertwus{v}{\exseq'}{\emptyseq}$\tcp*{insert $v$ into the wakeup tree at $E'$}\label{algacsl:event-insert}\label{algacsl:event-race-end}
        }
      }
    }
  }
\Else(\label{algacsl:event-explore-begin}\tcp*[f]{If not at a maximal execution sequence, explore...}){
  \If{$\wut{\exseq} = \emptytree$ }{\label{algacsl:exploration-begin}
    $\keyword{choose}\; p \in \enabled{\exseq}$\tcp*{... or by selecting an arbitrary $p$...}\label{algacsl:wut-empty-choose}
    $\wut{\exseq} := \tuple{\set{\emptyseq,p},\set{(p,\emptyseq)}}$\tcp*{Adapt wakeup tree accordingly}\label{algacsl:wut-empty-init}
  }
  \ForEach{message $q$ that is active after $\exseq$}{$\accesses(q) := \emptyset$\tcp*{Initialize the sequences of accesses for messages}\label{algacsl:accesses-initialize}}
  \While(\tcp*[f]{While the wakeup tree is not empty...}){$\exists q \in \wut{\exseq}$\label{algacsl:while-WuT}}{
    $\keyword{let}\; p = \min_{\prec} \{ q \in \wut{\exseq} \}$\label{algacsl:pick-WuT}\tcp*{... pick next branch, ...}
    $\wut{\exseq.p} := \subtreeafter{p}{\wut{\exseq}}$\label{algacsl:def-WuT}\tcp*{extract next wakeup tree)}
    $\keyword{let}\; \tmpaccesses = \explore(\exseq.p)$\label{algacsl:event-call-explore}\tcp*{... and make a recursive call}
    \If{$\nextev{E}{p}$ is the last event of message $p$}{\label{algacsl:collection-start}
      add $p$ to $\dom{\tmpaccesses}$ with $\tmpaccesses(p) = \set{\emptyseq}$ \label{algacsl:initialize-accesses}
    }
    \If{$\nextev{E}{p}$ performs a global access}{
      prepend $\nextev Ep$'s access to each sequence in $\tmpaccesses(p)$\label{algacsl:extend-tmpaccesses}
    }
    \ForEach{message $q$ that is active after $\exseq$}{$\accesses(q) \ \cup\!= \tmpaccesses(q)$}\label{algacsl:accumulate-accesses}
    add $p$ to the domain of $\done(\exseq)$\label{algacsl:doneset-add}\tcp*{Mark $p$ as explored}
    \If(\tcp*[f]{If $p$ starts}){$p$ starts after $E$}{
      $\done(\exseq)(p) := \accesses(p)$\tcp*{... store $p$'s accesses}\label{algacsl:donesleeptree-add}
    }
    \lElse(\tcp*[f]{... store $\nextev{E}{p}$'s access}){$\done(\exseq)(p) := \mbox{$\nextev{E}{p}$}$'s access}\label{algacsl:donesleeptree-add-ev}
    {remove all sequences of form $p.w$ from $\wut{\exseq}$}\tcp*{At end, cleanup}\label{algacsl:wut-deletebranch}
  }
  $\keyword{return}\; \accesses$\label{algacsl:exploration-end}
}
}
\caption{\EventDPOR}
\label{alg:eventdpor-access}
\end{algorithm}

The exploration phase (\crefrange{algacsl:exploration-begin}{algacsl:exploration-end})
is entered if exploration has not reached the end of a maximal execution sequence.
First, if $\wut{\exseq}$ only contains the empty sequence, then 
an arbitrary enabled message is entered into $\wut{\exseq}$ (\cref{algacsl:wut-empty-choose,algacsl:wut-empty-init}).
Thereafter, each sequence in $\wut{\exseq}$ is subject to recursive exploration.
We find the $\prec$-minimal child $p$ of the root of $\wut{\exseq}$ (\cref{algacsl:pick-WuT}),
and make the recursive call $\explore(\exseq.p)$ (\cref{algacsl:event-call-explore}).
Before the call, $\wut{\exseq.p}$ is initialized (\cref{algacsl:def-WuT}).
During the call $\explore(\exseq)$, information is also collected about the sequences of shared-variable accesses that can be performed by each message that is active after $\exseq$,
and subsequently stored in the mapping $\done$.
The information is collected in the variable $\accesses$, which
is initialized at \cref{algacsl:accesses-initialize}.
Each recursive call $\explore(\exseq.p)$ returns the sets of access sequences performed by messages that are active after $\exseq.p$ (\cref{algacsl:event-call-explore}).
After prepending the access performed by $\nextev{E}{p}$ to the sets of access sequences performed by $p$ (\cref{algacsl:extend-tmpaccesses}),
the sets  returned by $\explore(\exseq.p)$ are added to the corresponding sets in $\accesses$ (\cref{algacsl:accumulate-accesses}).
Finally, $p$ is added to the domain of $\done(E)$ (\cref{algacsl:doneset-add}).
If $p$ starts a message after $E$, then  $\done(\exseq)(p)$ is assigned the set of access sequences performed by $p$ (\cref{algacsl:donesleeptree-add}), otherwise only the access of
$\nextev{E}{p}$.
Thereafter, the subtree rooted at $p$ is removed from $\wut{\exseq}$ (\cref{algacsl:exploration-end}).
When all recursive calls of form $\explore(\exseq.p)$ have returned, the accumulated sets of access sequences are returned (\cref{algacsl:exploration-end}).

\EventDPOR  calls functions that are briefly described in the following paragraphs. 
More elaborate descriptions (with pseudocode) are in~\cref{sec:functions-appendix}.

\noindent\emph{$\reverserace({E},{\event},{\event'})$} is given a race $\event \revrace{\exseq} \event'$
in the execution $\exseq$ (\cref{algacsl:race-loop}),
and returns a set of executions that reverse the race in the sense that they 
perform the second event $e'$  of the race without performing the first one, and (except for $e'$) only contain events that are not affected by the race.
More precisely, it returns a set of pairs of form $\pair{E'}{u.e'}$, such that
\begin{inparaenum}[(i)]
\item
$E'.u$ is a maximal happens-before prefix of $E$ such that $E'.u.e'$ is an execution, and
\item $\dom{E'}$ is a maximal subset of $\dom{E'.u}$ such that $E' \leq E$.
\end{inparaenum}
An illustration of the $\reverserace$ function was given for the race on \texttt{x} in the program of \cref{fig:example3}.

  \hasbeenremoved{
    As an illustration, consider the race on \texttt{x} in the program of \cref{fig:example3}.
Here, there is a unique (up to equivalence) maximal execution which reverses the race, which consists of all events that post messages, all events in messages
    $p_2$ and $q_2$, and the assignment to \texttt{d} by $p_1$. The read of \texttt{x} by~$q_2$ should be ordered last, since it corresponds to the racing event $e'$.
    Message $q_1$ is removed by the rule at \cref{algl:revrace-rule-must}, whereby also the second of event of $p_1$ is removed, since it reads from the first event in $q_1$.
  }
  
\noindent\emph{$\insertwus{v}{E'}{\emptyseq}$} inserts the wakeup sequence $v$ into the wakeup tree $\wut{E'}$. If there is already some sequence $u$ in $\wut{E'}$ such that $u \mtprefixafter{E'} v$ or
  $v \mtprefixafter{E'} u$, then the insertion leaves $\wut{E'}$ unaffected. Otherwise
  $\insertwus{v}{E'}{\emptyseq}$ attempts to find the $\prec$-minimal non-leaf sequence $u$ in $\wut{E'}$ with $u \mtprefixafter{E'} v$, and
  insert a new leaf of form $u.v'$ into $\wut{E'}$, such that $v \mtprefixafter{E'} u.v'$, which is ordered after all existing descendants of $u$ in $\wut{E'}$.
  The function finds such a $u$ by descending into $\wut{E'}$ one event at a time; from each node $u'$ it finds a next node $u'.p$ as the $\prec$-minimal child with 
  $u'.p \mtprefixafter{E'} v$. If, during this search, the message $p$ starts after $E'.u'$ it may happen that the wakeup tree does not contain enough subsequent events to determine whether
  $u'.p \mtprefixafter{E'} v$; in this case the sequence $v$ is ``parked'' at the node $u'.p$: the insertion of $v$ will be resumed when $E'.u'.p$ is extended to a maximal execution (at \cref{algacsl:insert-pendingwu} with $E'$ being $E'.u'$).

\noindent\emph{$\insertpendingwu{v}{E'}$} inserts a wakeup sequence $v$, which is parked after a prefix $E'$ of the execution $E$, into an appropriate wakeup tree.
The function first decomposes $E'$ as $E''.p$, and checks whether
$p \in \winits{\exseq''}{v}$, using 
information about the accesses of $p$ that can be found in $E$.
If the check succeeds, then insertion proceeds recursively
one step further in the execution $E$, otherwise $v$ conflicts with $p$ and should be inserted into the wakeup tree after $E''$.

\noindent\emph{Checking for Redundancy}
Tests of form $p \in \winits{\exseq}{w}$ for a message $p$ and an execution $\exseq.w$
appear at
\cref{algacsl:event-test} and in the functions $\insertwusname$ and $\insertpendingwuname$.
If $p$ does not start after $E$, then the check can be straightforwardly performed using
sleep sets~\cite{Godefroid:thesis}.
If $p$ starts after $E$, then checking whether $p \in \winits{\exseq}{w}$ is NP-hard in the general case (see~\cref{thm-lowerbound}).
To avoid expensive calls to a decision procedure,
\EventDPOR employs a sequence of incomplete checks, starting with simple ones, and proceeding with a next test only if the preceding was not conclusive.
These tests are in order:
\begin{inparaenum}[1)]
  \item
  If $p$ is the first message (if any) on its handler in $w$, then $p \in \winits{\exseq}{w}$ is trivially true.
\item
  If the happens-before relation precludes $p$ from executing first on its handler, then $p \in \winits{\exseq}{w}$ is false; checking this may require
  $w$ to be extended so that $p$ (and possibly other messages) are executed to completion.
\item An attempt is made to construct an actual execution in which $p$ is the first message on its handler, which respects the happens-before ordering.
\item If all previous tests were inconclusive, a decision procedure is invoked as a final step.
\end{inparaenum}

\hasbeenremoved{\subsection{Reversing Races}
\label{sec:race-reversals}
\begin{algorithm}[t]
\SetAlFnt{\small\sf}
\BlankLine
\Fn{$\reverserace({E},{\event},{\event'})$}{
  \keyword{let} $E''$ be the subsequence of $E$ consisting of the events $e'''$ with $\event \nhappbf{\hb}{E} \event'''$\;\label{algl:revrace-short-init-notdep}
  \keyword{let} $S$ be a set of pairs of form $\pair{E'}{u.e'}$ which \\
  \quad for each maximal subset of $E''\cup\set{e'}$ that can be linearized \\
  \quad \quad \qquad \qquad to an execution of form $E'.u.e'$ while respecting $\happbf{hb}{E''}$ \\
  \quad contains one such execution $E'.u.e'$ where $E'$ is a maximal prefix of $E$.

  $\Return(S)$\;
}
\caption{Reversal of a Race.}
\label{alg:reverserace-short}
\end{algorithm}

A key procedure of \EventDPOR is $\reverserace$ which constructs new executions by analyzing and reversing a race in an explored execution.
This procedure is given a race $\event \revrace{\exseq} \event'$ in the currently explored execution $\exseq$ (at \cref{algacsl:race-loop} of \cref{alg:eventdpor-access}),
and returns a set of maximal executions that reverse the race in the sense that they 
perform the second event $e'$  of the race without performing the first one, and (except for $e'$) only contain events that are not affected by the race.
The procedure $\reverserace({E},{\event},{\event'})$, shown in \cref{alg:reverserace-short}, returns, 
for each maximal subset of $E''\cup\set{e'}$ that can be linearized to an execution that ends in $e'$, one such execution $E'.u.e'$ in which
$E'$ is a maximally long prefix of $E$; each execution $E'.u.e'$ is returned as a pair of form $\pair{E'}{u.e'}$.
}

\hasbeenremoved{
The \EventDPOR algorithm maintains the following properties
\begin{itemize}
   \item[P1:]
  whenever the exploration of some subtree rooted at some execution $\exseq.p \in \exseqs$ has completed,
  then for each maximal execution of form $\exseq.p.w$, the algorithm has explored an execution equivalent to $E.p.w$.
   \item[P2:]
whenever the exploration tree $\exseqs$ contains a node of form $\exseq.p$, then the algorithm will not add an execution of form $\exseq.w$ which is contained in
some execution of form $\exseq.p.w'$ for some $w'$, i.e., for which $p \in \winits{\exseq}{w}$.
\end{itemize}
Property P1 is the basic property which guarantees correctness. It is also used to avoid redundant exploration by enforcing P2. Such a check for redundancy is performed
before inserting a new wakeup sequence (\cref{algacsl:event-test}),
and also inside the procedure for wakeup tree insertion (\cref{alg:wakeuptree}).
}

\section{Correctness and Optimality}
\label{sec:correctness}
A program is defined to be \emph{non-branching} if each
message, which executes on the same handler as  some other message, performs the same sequence of accesses (reads or writes) to shared
variables during its execution, regardless of how its execution is interleaved with other threads and messages.
Note that the ``non-branching'' restriction does not apply to non-handler threads nor to messages that are the only ones executing on their handler.

The following theorems state that \EventDPOR is 
\emph{correct} (explores at least one execution in each equivalence class)
for \emph{all} event-driven programs
and \emph{optimal} (explores exactly one execution in each equivalence class)
for non-branching programs.
Proofs can be found in \cref{sec:correctness-proof}.

\hasbeenremoved{
Let us now
assume a particular completed execution of \EventDPOR. This execution 
consists of a number of terminated calls to $\explore(E)$ for some values 
of the parameters $E$ and $\WuT$. Let $\exseqs$ denote the set of execution 
sequences $E$ that have been explored in some call $\explore(E)$. Define 
the ordering $\treeorder$ on $\exseqs$ by letting $E \treeorder E'$ if 
$\explore(E)$ returned before $\explore(E')$. Intuitively, if one 
were to draw an ordered tree that shows how the exploration has proceeded, then 
$\exseqs$ would be the set of nodes in the tree, and $\treeorder$ would be the 
post-order between nodes in that tree. The correctness and optimality of 
\cref{alg:eventdpor-access} are stated in the following theorems.
}

\begin{theorem}[Correctness]
\label{thm:correctness}
Whenever the call to $\explore(\emptyseq)$ returns during \cref{alg:eventdpor-access},
then for all maximal execution sequences $E$, the algorithm has explored
some execution sequence in~$\eqclass{E}$.
\end{theorem}

\hasbeenremoved{
Since the initial call to the algorithm, $\explore(\emptyseq)$, starts with
the empty sequence as argument, \cref{thm:correctness}
implies that for all maximal execution sequences $E$ the algorithm
explores some execution sequence $E'$ which is in $\eqclass{E}$.
  Note also that a sequence of form $E.w$ need not have been explored inside
the call $\explore(E)$, but can have been explored in some earlier call,
of form $\explore(E'.p)$ for some prefix $E'$ of $E$.
}

\begin{theorem}[Optimality]
\label{thm:optimality}
When applied to a non-branching program,
\cref{alg:eventdpor-access} never explores two maximal execution
sequences which are equivalent.
\end{theorem}

\hasbeenremoved{
  \Cref{thm:optimality} ensures
that all maximal execution sequences reached are non-redundant.
}


\section{Implementation} \label{sec:impl}
\EventDPOR was implemented on top of \Nidhugg.
\Nidhugg~\cite{tacas15:tso} is a state-of-the-art stateless model checker for
C/C++ programs with Pthreads, which works at the level of the LLVM Intermediate
Representation.
\Nidhugg comes with a selection of DPOR algorithms. One of them is
\OptimalDPOR, which we have used as a basis for \EventDPOR's implementation.

We have extended the data structures of \Nidhugg with the information
needed by \EventDPOR.
For instance, nodes in wakeup trees contain new information, such as the set
of parked wakeup sequences, and events in executions include the information in
$\tmpaccesses$, used to compute the \done set
as shown in \crefrange{algacsl:initialize-accesses}{algacsl:donesleeptree-add}
of~\cref{alg:eventdpor-access}.
The relation $\happbf{\hb}{E}$ is represented by a
vector clock per event, containing the set of preceding events.
When reversing races (in $\reverserace$) and checking for redundancy
(\cref{algacsl:event-test} of~\cref{alg:eventdpor-access}),
the relation $\happbf{\hb}{E}$ is extended by a saturation operation
(\cref{def:weakall} in \cref{sec:functions-appendix}) that captures ordering constrained induced by serialized message execution.

Concerning race reversal, instead of reversing multiple races between messages
executed on the same handler, our implementation detects and reverses only
the race induced by the first conflict, since other races cannot be reversed,
as explained using the example in \cref{fig:example1new}.
Moreover, in cases where $\reverserace$ would return several maximal executions that reverse a race, our implementation instead returns their union, even though it may not form an execution (e.g., since it may contain several incomplete executed messages on a handler). From this union, events will be removed adaptively during wakeup tree insertion to extract only those maximal executions that generate new leaves in a wakeup tree.

\hasbeenremoved{Moreover, we postpone the deletion of incomplete messages from $\exseq''$, described at
\crefrange{algl:converge-begin}{algl:converge-end} in \cref{alg:reverserace},
to the point before adding a new branch in the wakeup tree.
This avoids unnecessary creation of wakeup sequences that will not create any
new leaf in the wakeup tree, since the deletion of incomplete messages
can be guided by the redundancy checks performed during the insertion.
Finally, instead of computing $\happbf{sc}{E.w'}$ for each test of form
$p \in \winits{E}{w}$ at \cref{algacsl:event-test}
of \cref{alg:eventdpor-access}, we precompute $\happbf{sc}{E.w}$
before all the tests, which accelerates the redundancy check.}

\section{Evaluation} \label{sec:eval}
In this section, we evaluate the performance of our implementation and put it
into context.
Since currently there is no other SMC tool for event-driven programs to
compare against,\footnote{All our attempts to use \href{https://github.com/eth-sri/R4}{$R^4$} failed miserably; the tool has not been updated since~2016.}
we have created an API, in the form of a C header file, that
implements event handlers as pthread mutexes (locks) and simulates messages as
threads that wait for their event handler to be free. This API allows us to
use plain C/pthread programs to compare \EventDPOR with the \OptimalDPOR
algorithm implemented in \Nidhugg as baseline, but also with the
\emph{Lock-Aware Partial Order Reduction} (\emph{LAPOR})
algorithm~\citet{LAPOR@OOPSLA-19}, implemented in \GenMC. The \LAPOR algorithm
is often analogous to \EventDPOR w.r.t. the amount of reduction that it can
achieve when event handlers are modeled as global locks.
We also include in our comparison the baseline DPOR algorithm of \GenMC that
tracks the modification order (\genmcmo{\small}) of shared variables.
%
For \Nidhugg, we used its \texttt{master} branch at the end of~2022;
%
for \GenMC, we used version 0.6.1.%
\footnote{\GenMC v0.6.1 (released July 2021) warns that \LAPOR usage with
\genmcmo{\footnotesize} is experimental; in fact, \LAPOR support has been
dropped in more recent \GenMC versions.}
%
We have run all benchmarks on a Ryzen 5950X desktop running Arch Linux.

\newcommand\SZ{\scriptsize}
\newcommand\error{\textcolor{red}{\SZ \textdagger}\xspace}
\newcommand\timeout{\textcolor{blue}{\SZ \clock}\xspace}
\newcommand\notavail{\textcolor{gray}{\SZ n/a}\xspace}
\newcommand\bug{\SZ\color{red} bug }

\pgfplotstableset{
  columns/benchmark/.style={string type,
    string replace*={_}{\_},
  },
  columns/tool/.style={string type},
  columns/lapormo_traces/.style={column name=\multicolumn{1}{r}{\lapormo{\SZ}}},
  columns/lapormo_time/.style={column name=\multicolumn{1}{r}{\lapormo{\SZ}}},
  columns/genmcmo_traces/.style={column name=\multicolumn{1}{r}{\genmcmo{\SZ}}},
  columns/genmcmo_time/.style={column name=\multicolumn{1}{r}{\genmcmo{\SZ}}},
  columns/optimal_traces/.style={column name=\multicolumn{1}{r}{\opt{\SZ}}},
  columns/optimal_time/.style={column name=\multicolumn{1}{r}{\opt{\SZ}}},
  columns/event_traces/.style={column name=\multicolumn{1}{r}{\evt{\SZ}}},
  columns/event_time/.style={column name=\multicolumn{1}{r}{\evt{\SZ}}},
  every head row/.style={before row={%
       \toprule
      & \multicolumn{4}{c}{Executions (Traces+Blocked)} & \multicolumn{4}{c}{Time (secs)}\\
      \cmidrule(r){2-5}\cmidrule(r){6-9}
      & \multicolumn{2}{c}{\GenMC} & \multicolumn{2}{c}{\Nidhugg} &
      \multicolumn{2}{c}{\GenMC} & \multicolumn{2}{c}{\Nidhugg} \\
      \cmidrule(r){2-3}\cmidrule(r){4-5}\cmidrule(r){6-7}\cmidrule(r){8-9}
    },after row=\midrule},
  every last row/.style={after row=\bottomrule},
  column type={r},
  create on use/Benchmark/.style={
    create col/assign/.code={%
      \getthisrow{benchmark}\benchmark
      \getthisrow{n}\n
      \edef\entry{\benchmark(\n)}%
      \pgfkeyslet{/pgfplots/table/create col/next content}\entry
    },
  },
  columns/Benchmark/.style={
    column type={l},
  },
  columns={Benchmark,
    genmcmo_traces,lapormo_traces,optimal_traces,event_traces,
    genmcmo_time,lapormo_time,optimal_time,event_time},
  fixed,
  string type, 
}

We will compare implementations of different DPOR algorithms based on the
number of executions that they explore, as well as the time that this takes.
For some programs, \LAPOR also examines a fair amount
of \emph{blocked} executions (i.e., executions that cannot
be serialized and need to be aborted), which naturally affects its time
performance. In \cref{tab:eval}, we show the number of executions explored by an
entry of the form $T$+$B$, where $T$ is the number of complete traces and $B$
is the number of blocked executions. (We omit the $B$ part when
it is zero.)

\begin{table}[t]
  \caption{Performance of different DPOR algorithm implementations.}
  \label{tab:eval}
  \centering\SZ
  \setlength{\tabcolsep}{2pt}
  \pgfplotstablevertcat{\output}{results/laban/posters.txt}
  \pgfplotstablevertcat{\output}{results/laban/buyers.txt}
  \pgfplotstablevertcat{\output}{results/laban/ping_pong.txt}
  \pgfplotstablevertcat{\output}{results/laban/consensus.txt}
  \pgfplotstablevertcat{\output}{results/laban/prolific.txt}
  \pgfplotstablevertcat{\output}{results/laban/sparse-mat.txt}
  \pgfplotstablevertcat{\output}{results/laban/plb.txt}
  \pgfplotstabletypeset[
    every row no 3/.style={before row=\midrule},
    every row no 6/.style={before row=\midrule},
    every row no 9/.style={before row=\midrule},
    every row no 12/.style={before row=\midrule},
    every row no 15/.style={before row=\midrule},
    every row no 18/.style={before row=\midrule},
  ]{\output}
\end{table}


All the benchmark programs we use are parametric, typically on the number of
threads used (and thus messages posted); their parameters are shown inside
parentheses.
In the first program (\bench{posters}), each thread posts to a single event
handler two messages containing stores to some atomic global variable, and
then the value of this variable is checked by an assertion. This simple
program allows us to establish the baseline speed of all implementations. We
can see that \GenMC~\genmcmo{\small} is the fastest here. The reason is that 
it does not perform any checks whether the explored executions are sequentially 
consistenct, which allows it to be five times faster than \LAPOR, and seven 
to nine times faster than \Nidhugg's algorithm implementations. We can also
notice that \EventDPOR incurs a small but noticeable overhead over
\OptimalDPOR for the extra machinery that its implementation requires.

The next two benchmarks were taken from a paper by Kragl et al.~\citet{Kragl20}.
In \bench{buyers}, $n$ ``buyer'' threads coordinate the purchase of an item
from a ``seller'' as follows: one buyer requests a quote for the item from the
seller, then the buyers coordinate their individual contribution, and finally
if the contributions are enough to buy the item, the order is placed.
In \bench{ping-pong}, the ``pong'' handler thread receives messages with
increasing numbers from the ``ping'' thread, which are then acknowledged back
to the ``ping'' event handler.

Looking at~\cref{tab:eval}, we notice that, in both \bench{buyers} and
\bench{ping-pong}, all algorithms explore the same number of traces, but
\LAPOR also explores a significant number of executions that cannot be
serialized and need to be aborted. In fact, for both benchmarks, the aborted
executions significantly outnumber the traces explored.  This affects
negatively the time that \LAPOR takes, and \GenMC \lapormo{\small} becomes the
slowest implementation.  In contrast, \EventDPOR does not suffer from this
problem and shows similar scalability as baseline \GenMC and \OptimalDPOR.

With the four remaining benchmarks, we evaluate all implementations in programs
where algorithms tailored to event-driven programming, either natively
(\EventDPOR) or which are lock-aware (when handlers are implemented as locks),
have an advantage.
The first program (\bench{consensus}), again from the paper by Kragl et
al.~\citet{Kragl20}, is a simple \emph{broadcast consensus} protocol for $n$
nodes to agree on a common value. For each node~$i$, two threads are created:
one thread executes a \texttt{broadcast} method that sends the value of
node~$i$ to every other node, and the other thread is an event handler that
executes a \texttt{collect} method which receives~$n$ values and stores the
maximum as its decision. Since every node receives the values of all other
nodes, after the protocol finishes, all nodes have decided on the same value.
%
%
The next program (\bench{prolific}) is synthetic: $n$ threads send $n$
messages with an increasing number of stores to and loads from an atomic
global variable to one event handler.
The \bench{sparse-mat} program computes the number of non-zero elements of a
sparse matrix of dimension $m \times n$, by dividing the work into $n$ tasks
sent as messages to different handlers, which compute and join their results.
The last benchmark (\bench{plb}) is taken from a paper by Jhala and
Majumdar~\citet{popl07:JhalaM}. A fixed sequence of task requests is received
by the main thread. Upon receiving a task, the main thread allocates a space
in memory and posts a message with the pointer to the allocated memory that
will be served by a thread in the future.

Refer again to~\cref{tab:eval}.
In \bench{consensus}, all algorithms start with the same number of traces, but
\LAPOR and \EventDPOR need to explore fewer and fewer traces than the other
two algorithms, as the number of nodes (and threads) increases. Here too,
\LAPOR explores a significant number of executions that need to be aborted,
which hurts its time performance. On the other hand, \EventDPOR's handling
of events is optimal here.
The \bench{prolific} program shows a case where algorithms not tailored to
events (or locks) explore $(n-1)!$ traces, while \LAPOR and \EventDPOR explore
only $2^n-2$ consistent executions, when running the benchmark with
parameter $n$. It can also be noted that \EventDPOR scales \emph{much} better
than \LAPOR here in terms of time, due to the extra work that \LAPOR needs to
perform in order to check consistency of executions (and abort some of them).
The \bench{sparse-mat} program shows another case where algorithms that are
not tailored to events explore a large number of executions unnecessarily
(\timeout denotes timeout). This program also shows that \EventDPOR beats
\LAPOR time-wise even when \LAPOR does not explore executions that need to be
aborted.
Finally, \bench{plb} shows a case on which \EventDPOR and \LAPOR really
shine. These algorithms need to explore only one trace, independently of the
size of the matrices and messages exchanged, while DPOR algorithms not
tailored to event-driven programs explore a number of executions which
increases exponentially and fast.

We remark that, in all benchmarks, the inexpensive checks for redundancy were
sufficient, and \EventDPOR explored the optimal number of traces.
Results from an extended set of benchmarks appear in~\cref{app:eval-complete}.


\section{Concluding Remarks}
\label{sec:conclusions}

In this paper, we presented a novel SMC algorithm, \EventDPOR, tailored to the
characteristics of event-driven multi-threaded programs running under the SC
semantics. The algorithm was proven correct and optimal for event-driven
programs in which the variable accesses of events do not depend on how their
execution is interleaved with other threads.

We have implemented \EventDPOR in the \Nidhugg tool, and we will open-source
our implementation.
With a wide range of event-driven programs, we have shown that \EventDPOR
incurs only a moderate constant overhead over its baseline implementation
(\OptimalDPOR), it is exponentially faster than existing state-of-the-art SMC
algorithms in time and number of traces examined on programs where events'
actions do not conflict, and does not suffer from performance degradation
caused by having to examine
non-serializable executions.
%

\EventDPOR assumes that handlers can process their events in arbitrary order.
Directions for future work include to retarget \EventDPOR for event-driven
programs with other policies (e.g., FIFO), and for specific event-driven
execution models.

\section{Reproducible Artifact}
An anonymous artifact containing the benchmarks and all the tools used in the evaluation, including our Nidhugg with Event DPOR, is available at
\url{https://doi.org/10.5281/zenodo.7929004}.

\bibliographystyle{splncs04}
\bibliography{./bibdatabase.bib}

\appendix
\clearpage
\section{Detailed Descriptions of Auxiliary Functions}
\label{sec:functions-appendix}
In this section, we describe in detail the functions that are called by \EventDPOR, and were briefly described at the end of \cref{sec:algo:access-sets}
Some of these functions extend the happens-before relation $\happbf{\hb}{E}$ on an execution with
additional ordering constraints that are enforced  in the event-driven execution model, stemming from the fact that
each handler must execute its messages in some sequential order. The following \emph{saturation operation} adds such additional orderings imposed by
any ordering relation on events.
\begin{definition}[Saturation]
  \label{def:weakall}
Let $E$ be a sequence of events, and $\happbf{}{E}$ be an irreflexive partial order on the events of $E$. 
We define
$\weaksatrel{\happbf{}{E}}$ as the smallest transitive relation $\happbf{\sat}{E}$ on the events of $E$ which includes
$\happbf{}{E}$ and satisfies the constraint that  whenever $e$ and $e'$ are events in different messages on the same handler,
and there is an event $e''$ in the same message as $e$ and an event $e'''$  in the same
  message as $e'$ with $e'' \happbf{\sat}{E} e'''$,
then $e \happbf{\sat}{E} e'$.
\qed
\end{definition}
In the above definition, note that it is not required that $e$ is distinct from $e''$, nor that $e'$ is distinct from $e'''$.

\subsection{Reversing Races}
\label{sec:race-reversals-appendix}
A key procedure of \EventDPOR is $\reverserace({E},{\event},{\event'})$ which constructs new executions by analyzing and reversing a race in an explored execution.
This procedure is given a race $\event \revrace{\exseq} \event'$ in the currently explored execution $\exseq$ (at \cref{algacsl:race-loop} of \cref{alg:eventdpor-access}),
and returns a set of maximal executions that reverse the race.
More precisely, it returns a set of pairs of form $\pair{E'}{u.e'}$, such that
\begin{inparaenum}[(i)]
\item
$E'.u$ is a maximal happens-before prefix of $E$ such that $E'.u.e'$ is an execution, and
\item $\dom{E'}$ is a maximal subset of $\dom{E'.u}$ such that $E' \leq E$.
\end{inparaenum}

The procedure $\reverserace({E},{\event},{\event'})$ is shown in \cref{alg:reverserace}. Let $E''$ be the set of events of $E$ that are not affected by the race (\cref{algl:revrace-init-notdep}):
this is the set of events $e'''$ with $\event \nhappbf{\hb}{E} \event'''$. If $E''$ can be reordered to form an execution, the code at
\crefrange{algl:revrace-init-ordering}{algl:converge-end} will have no effect;
$\reverserace$ will terminate and returns its linearization.
However, there are situations in which $E''$ cannot be reordered into an execution.
For instance, $E''$ may contain two incomplete messages on the same handler because the remaining parts of these messages happen-after $e$ in $E$.
Since an execution may contain at most one incomplete message per handler,
$\reverserace$ then performs a sequence of message removals and reorderings to produce a set of maximal wakeup sequences.
  The procedure employs the saturation operation of \cref{def:weakall} to constrain the ordering between messages on the same handler. 
  The procedure maintains
  an ordering relation $\happbf{sc}{E''}$ on $E''$, initialized to $\weaksatrel{\happbf{hb}{E''}}$ (\cref{algl:revrace-init-ordering}).
It thereafter performs a sequence of steps in which messages are removed from $E''$ and/or the ordering relation $\happbf{sc}{E''}$ is extended. Some steps may be resolved nondeterministically: in such cases the procedure pursues all possible alternatives, potentially resulting in several returned sequences. The steps of \cref{alg:reverserace} are the following.
\begin{algorithm}[!htp]
\SetAlFnt{\small\sf}
\BlankLine
\Fn{$\reverserace({E},{\event},{\event'})$}{
  \keyword{let} $E''$ be the subsequence of $E$ consisting of the events $e'''$ with $\event \nhappbf{\hb}{E} \event'''$\;\label{algl:revrace-init-notdep}
  \keyword{let} $\event''$ be the last event performed by $\procof{\event'}$\;
  $\happbf{sc}{E''} := \weaksatrel{\happbf{hb}{E''}}$\;\label{algl:revrace-init-ordering}
  \keyword{let} $S=\{E''\}$\label{algl:revrace-lin-seq}\;
  \ForEach{$F\in S$}{
    \Repeat{each handler in $F$ has at most one incomplete message}{\label{algl:revrace-msg-del-start}
      \label{algl:converge-begin}
      \If{an incomplete message includes an event $e'''$ with $e''' \happbf{\hb}{F} e''$}{\label{algl:revrace-rule-must}
        remove all other incomplete messages on same handler from $F$\;
      }
      \If{several incomplete messages execute on one
        handler}{\label{algl:revrace-rule-choose}
        \ForEach{incomplete message $p$ in the same handler}{
          construct a sequence $U$ where all the messages except $p$ from
          this handler are removed\;
          add $U$ to the set $S$\;
        }
        delete $F$ from $S$\;
        pick another sequence $F$ from $S$\;
      }
      \ForEach{incomplete message $p$}{\label{algl:revrace-rule-last}
        add relation $\happbf{sc}{F}$ from all other messages on
        same handler to $p$ and saturate\;
      }
      \ForEach{cycle in $\happbf{sc}{F}$}{\label{algl:revrace-rule-cycle}
        remove a message in the cycle\;
      }
      remove events that follow already removed events in the $\happbf{\hb}{F}$ ordering\;
      \If{an event $e'''$ s.t. $e''' \happbf{\hb}{F} e''$ is
        deleted}{\label{algl:revrace-irreversible}
        remove $F$ from $S$ and exit the loop;
      }
    }\label{algl:converge-end}
  }\label{algl:revrace-msg-del-end}
  \keyword{let} $WSS=\emptyset$ \tcp*[f]{set of wakeup sequences}\;
  \ForEach{$F\in S$}{
    \ForEach{two messages from the same handler that are not ordered by $\happbf{sc}{F}$}{
      add relation $\happbf{sc'}{F}$ as they appear in $F$; \label{algl:revrace-msg-order}
    }
    \Repeat{until $\happbf{sc}{F}$ and $\happbf{sc'}{F}$ together
      are acyclic}{
      nondeterministically pick two messages ordered by the relation $\happbf{sc'}{F}$ and
      reverse the order; \label{algl:revrace-next-msg-order}
    }
    topologically sort $F$ respecting
    $\happbf{sc}{F}$ and $\happbf{sc'}{F}$ \label{algl:revrace-linearize}\;
    extract largest common prefix $E'$ of $F$ and $E$\;
    add \tuple{E',.u.e'} to $WSS$, where $F=E'.u$\;
  }
  $\Return(WSS)$\;
}
\caption{Reversal of a Race.}
\label{alg:reverserace}
\end{algorithm}

\begin{description}
  \item[\cref{algl:revrace-lin-seq}]
    After the loop from \cref{algl:revrace-msg-del-start}
    to \cref{algl:revrace-msg-del-end}, the set $S$ will contain all
    the possible sequences with at most one incomplete message per handler.
  \item[\cref{algl:revrace-rule-must}]
    If an incomplete message includes an event $e'''$ with $e''' \happbf{hb}{F} e''$, then
    any other message on the same handler which is not completely executed in $F$ must be removed.
  \item[\cref{algl:revrace-rule-choose}]
    If several incomplete messages execute on the same handler, then
    finds all the possible sequence where only one of the incomplete
    messages is present and saves them to $S$.
    
  \item[\cref{algl:revrace-rule-last}]
    Whenever a handler has an incomplete message $p$, any other message $p'$
    on that handler must be executed before $p$, represented by extending $\happbf{sc}{F}$ from the last
    event of $p'$ to the first event of $p$ and then saturating.
  \item[\cref{algl:revrace-rule-cycle}]
    If $\happbf{sc}{F}$ becomes cyclic during the filtering and ordering procedure, then each cycle should be broken by
    removing the events in a suitable message.
  \item[\cref{algl:revrace-irreversible}]
    It is possible to have two or more incomplete messages from the same handler in $F$ each having at least one
    event that happens-before $e''$. Because of this reason or
    non-deterministic choice during message deletion process described
    previously, an event $e'''$ such that
    $e''' \happbf{\hb}{F} e''$ might be deleted from
    $\exseq''$. Then the algorithm removes $F$ from $S$.
  \item[\cref{algl:revrace-msg-order}]
    By adding additional relation $\happbf{sc'}{F}$,
    the algorithm determines a total order on the messages from the same handler.
  \item[\cref{algl:revrace-next-msg-order}]
    If $\happbf{sc}{F}$ and $\happbf{sc'}{F}$ together form a
    cycle, the algorithm tries to guess another order
    $\happbf{sc'}{F}$. Systematic search of $\happbf{sc'}{F}$ is a NP-complete
    problem in general case (see \cref{thm-consistency} below). But for the
    programs we have tried so far, doing \cref{algl:revrace-msg-order}
    is sufficient.
  \item[\cref{algl:revrace-linearize}]
    The sequence $u$ is linearized by topological sort procedure while
    respecting $\happbf{sc}{F}$ and $\happbf{sc'}{F}$.
\end{description}

As an illustration, consider the race on \texttt{x} in the program of \cref{fig:example3}.
Here, there is a unique (up to equivalence) maximal execution which reverses the race, which consists of all events that post messages, all events in messages
    $p_2$ and $q_2$, and the assignment to \texttt{d} by $p_1$. The read of \texttt{x} by~$q_2$ should be ordered last, since it corresponds to the racing event $e'$.
    Message $q_1$ is removed by the rule at \cref{algl:revrace-rule-must}, whereby also the second of event of $p_1$ is removed, since it reads from the first event in $q_1$.

\paragraph{\bf Event-driven Consistency.}
When describing \cref{algl:revrace-next-msg-order} above, we stated that the problem of determininig whether a given happens-before relation
can be obtained from some execution is NP-complete. This follows from NP-completeness of the event-driven consistency problem.
The event-driven consistency problem consists in checking whether, for a given  directed graph 
$(S,\happbf{\hb}{S})$ where $S$ is a set of events and $\happbf{\hb}{S}$ is a set of edges,  there is an execution sequence $E$ such that $(S,\happbf{\hb}{S})$ is the  \emph{\hb-trace} of  $E$.

\begin{theorem}
\label{thm-consistency}
The event-driven consistency problem is NP-complete.
\end{theorem}

The proof of the above theorem can be found in \cref{sec:complexity-proof-consistency}.  Given this NP-hardness result,  we define a procedure  to reverse races (\cref{sec:race-reversals-appendix}) that makes use of a saturation procedure to constrain the ordering between messages and therefore reduces the number of cases to  consider.


\subsection{Wakeup Tree Insertion}
\label{sec:wakeuptrees}
In this section, we formally define wakeup trees, and present the procedure $\insertwusname$ for inserting wakeup sequences, and $\insertpendingwuname$ for inserting parked wakeup sequences.


\begin{definition}[Wakeup Tree]
\label{def:Wut}
A \emph{wakeup tree} is an ordered tree $\tuple{B,\prec}$, where
$B$ (the set of \emph{nodes}) is a finite prefix-closed set of sequences of
messages, with the empty sequence $\emptyseq$ being the root.
The children of a node $u$, of form
$u.p$ for some set of messages $p$, are ordered by $\prec$.
In the tree $\tuple{B,\prec}$, such an ordering between children is
extended to a total order $\prec$ on~$B$ by letting
$\prec$ be the induced post-order relation between the nodes in $B$
(i.e., if the children $u.p_1$ and $u.p_2$ are ordered as
$u.p_1 \prec u.p_2$, then $u.p_1 \prec u.p_2 \prec u$ in the induced post-order).
\qed
\end{definition}

\begin{algorithm}[t]
\newcommand{\lIfElse}[3]{\lIf{#1}{#2 \textbf{else}~#3}}
\SetAlFnt{\small\sf}
\SetKw{Continue}{continue}
\BlankLine
\Fn{$\insertwus{v}{E'}{u}$}{
  \lIf{$v$ is the empty sequence \keyword{or} $u$ is a leaf in $\wut{E'}$}{\Return{}}\label{alg:wut-insert-empty}
  \ForEach{\mbox{\rm child $u.p$ of $u$, in order (from left to right)}}{\label{algl:wut-foreach-child}
    \If(\tcp*[f]{If a new message is not started ...}){$p$ does not start after ${\exseq'.u}$}{\label{alg:no-message}
      \If{$p \in \winits{\exseq'.u}{v}$}{
        \lIfElse{$\nextev{\exseq'.u}{p} \in v$}{$\insertwus{v\remove p}{E'}{u.p}$}{\Return{}\label{alg:travers-no-message}}
      }
    }
    \Else{\label{alg:message}
      \If{$p$ is the first (if any) message on its handler in $v$}{
        \lIfElse{$\nextev{\exseq'.u}{p} \in v$}{$\insertwus{v\remove p}{E'}{u.p}$}{\Return{}\label{alg:wut-pfirst}}
      }
      \ElseIf{$p$ is fully present in $v$\label{alg:wut-pfull}} {
	      \lIfElse{$p \in \winits{\exseq'.u}{v}$}{$\insertwus{v\remove p}{E'}{u.p}$}{
           \Continue
         }\label{alg:wut-pfull-continue}
      }
      \Else{
        add $v$ to $\pendingwus{E'.u.p}$\;
        \Return{}\;
      }
    }
  }
insert $v$ as a new branch from $u$, ordered after the existing children of $u$\;\label{alg:wut-insert-branch}
  \Return{}\;
}
\caption{Insertion into Wakeup Tree.}
\label{alg:wakeuptree}
\end{algorithm}

\hasbeenremoved{The insertion of a new wakeup sequence $v$ into a wakeup tree $\wut{E'}$ should enforce property P2.
The insertion of a new wakeup sequence $v$ into a wakeup tree $\wut{E'}$ is performed by
descending step-by-step from the root of the tree. At each node $u$, it
should, for each of its children (of form $u.p$), test whether $p \in\winits{E'.u}{v'}$, where
$v'$ is obtained from $v$ by removing the events that are already in $u$.
If $p \in\winits{E'.u}{v'}$ then insertion moves to the child $u.p$. If the test fails for all
children, then $v'$ should be added as a right-most child of $u$.
}



Insertion of a wakeup sequence $v$ into the wakeup tree $\wut{E'}$ is performed by calling the function $\insertwus{v}{E'}{u}$ with parameters $v$ and $E'$, and the parameter $u$ being the empty sequence.
The call $\insertwus{v}{E'}{\emptyseq}$ will, if $v$ conflicts with all its current leaves, extend the wakeup tree $\wut{E'}$ by a new leaf $v'$ such that $v' \mtequiv{E'} v$.
The recursive function $\insertwus{v}{E'}{u}$, shown in \cref{alg:wakeuptree},
traverses the wakeup tree $\wut{E'}$ from the root downwards, where $u$ is the current point of the traversal.
The initial call is performed with $u$ being the empty sequence.
Each invocation of $\insertwus{v}{E'}{u}$
first checks whether a leaf has been reached or all of $v$ has already been examined, in which case nothing new should be added to $\wut{E'}$ (\cref{alg:wut-insert-empty}).
Thereafter, it considers the children of $u$ (of form $u.p$) from left to right.
For each child $u.p$, the algorithm tries to determine whether or not $p \in \winits{\exseq.u}{v}$.
If $p$ does not start after $E'.u$ then $p \in \winits{\exseq.u}{v}$ then $p \in \winits{\exseq.u}{v}$ can be checked by simple inspection at \crefrange{alg:no-message}{alg:travers-no-message} (as described in the second paragraph of \cref{sec:checkwi}).
The algorithm traverses to $u.p$ by a call to $\insertwus{v\remove p}{E'}{u.p}$ if $p \in \winits{\exseq.u}{v}$, otherwise it
considers the next child of $u$ if $p \not\in \winits{\exseq.u}{v}$.
If $p \in \winits{\exseq.u}{v}$ but $p$ does not appear in $v$, then actually no wakeup sequence need be inserted (\cref{alg:travers-no-message}).
If $p$ starts after $E'.u$ (\cref{alg:message}), then
\begin{itemize}
\item the case in which  $p$ is the first (if any) message on its handler in $v$, considered at \cref{alg:wut-pfirst} is performed according to the \textbf{Simple Check} in \cref{sec:checkwi};
\item if $p$ executes to completion in the sequence $v$ (\cref{alg:wut-pfull}), then $v$ contains sufficient information to decide whether $p \in \winits{\exseq.u}{v}$ using the remaining sequence of checks in \cref{sec:checkwi};
\item if none of these two cases apply, then more information is needed about which accesses $p$ performs when it is executed. Therefore the sequence $v$ is ``parked'' at the node $u.p$: the insertion of $v$ will be resumed when the node $u.p$ is extended to a maximal execution starting with $E'.u.p$, which happens at \cref{algacsl:insert-pendingwu} of \cref{alg:eventdpor-access} with $E'$ being $E'.u$.
\end{itemize}
If all children $u.p$ of $u$ have been traversed with failing tests for $p \in \winits{\exseq.u}{v}$, then $v$ is added as a new branch from $u$, ordered after the already existing children (\cref{alg:wut-insert-branch}).

\begin{algorithm}[t]
\SetAlFnt{\small\sf}
\BlankLine
\Fn{$\insertpendingwu{v}{E'}$}{
  $\keyword{let} \ E''.p  \ \mbox{be} \ E'$\;
  \lIf{$v$ is the empty sequence}{\Return{}}
  \lIf{$\exseq''$ was formerly a leaf in a wakeup tree}{\Return{}}\label{alglpwu:former-leaf}
  \If{\label{algl:pwu-if-in-wi}$p \in \winits{\exseq''}{v}$}{
    $\keyword{let} \ q \ \mbox{be the message following $E'$ in $E$}$\;
    $\insertpendingwu{v\remove p}{E'.q}$\;\label{alglpwu:recursive-call}
   }
   \Else{
     $\insertwus{v}{\exseq''}{\emptyseq}$\tcp*{If conflict w. next event, insert into the wakeup tree}\label{alglpwu:event-insert}
   }
}
\caption{Insertion of Parked Wakeup Sequences.}
\label{alg:pendingwus}
\end{algorithm}

It remains to define the procedure for inserting parked wakeup sequences (called at~\cref{algacsl:insert-pendingwu} of~\cref{alg:eventdpor-access}).
This insertion is described in~\cref{alg:pendingwus}, as the function
$\insertpendingwu{v}{E'}$, which inserts a wakeup sequence $v$ which is parked after a prefix $E'$ of the execution $E$.
The function first decomposes $E'$ as $E''.p$, and checks whether
$p \in \winits{\exseq''}{v}$.
Information about the accesses of $p$ can now be found in the execution $E$, so that the check $p \in \winits{\exseq''}{v}$ can be performed.
The check will be exact for non-branching programs, but possibly conservative in general.
If the check succeeds, then insertion proceeds one step further in the execution $E$ (\cref{alglpwu:recursive-call}), otherwise $v$ conflicts with $p$ and so should be inserted at the wakeup tree after $E''$ (\cref{alglpwu:event-insert}).
As an additional optimization, \Cref{alglpwu:former-leaf} checks whether $E''$ was the leaf that is extended to the currently explored execution.
If so, the insertion can return without inserting anything, in analogy
with how leaves are handled in wakeup tree insertion
(\cref{alg:wut-insert-empty} of \cref{alg:wakeuptree}).

\subsection{Checking for Redundancy}
\label{sec:checkwi}
Let us now consider the problem of deciding whether $p \in \winits{\exseq}{w}$ for a message $p$ and an execution $\exseq.w$.

If $p$ does not start after $E$, then $p \in \winits{\exseq}{w}$ can be checked by simple inspection, as follows.
    If $\nextev{E}{p}$ is a local event or posts a message, then $p \in \winits{\exseq}{w}$ holds trivially.
  If $\nextev{E}{p}$ accesses a shared variable, then
  \begin{inparaenum}[(i)]
    \item if $p$ appears in $w$, we have $p \in \winits{\exseq}{w}$ precisely when there is no event $\event$ in $w$ such that
      $\event \happbf{hb}{E.w}{\nextev{E}{p}}$, and
    \item if $p$ does not appear in $w$, we have $p \in \winits{\exseq}{w}$ precisely when no event in $w$ conflicts with $\nextev{E}{p}$.
  \end{inparaenum}

If $p$ starts after $E$, then checking whether $p \in \winits{\exseq}{w}$
is NP-hard in the general case, as we show in \cref{thm-lowerbound}.
However, in many cases, the check can be performed by tests that run in polynomial time.
\EventDPOR employs the following sequence of checks, starting with simple ones, and resorting to an exact decision procedure only as a last step.
We assume that the event which posts message $p$ appears in $E$, otherwise
$p \in \winits{\exseq}{w}$ is trivially false.




\begin{description}
  \item[Simple Check]
  If $p$ is the first message (if any) on its handler in $w$, then $p \in \winits{\exseq}{w}$ is trivially true (recall our assumption that the first event of a message does not access a shared variable).
\item[Happens-Before Check]
  If $p$ is not the first message on its handler in $w$, 
  we check whether there is a happens-before dependency from a message $p'$ which precedes $p$ on its handler, as follows.
  \begin{enumerate}
  \item
   If $p$ is not executed to completion in $w$, we extend $w$ by a sequence of events performed by $p$ which performs all the shared-variable
  accesses that $p$ did not perform in $w$. If after this extension, some event of $p$ happens-after an event in a message $q$ on another handler which is
  not executed to completion in $w$, then $w$ is further extended by events of $q$ in the same way. If an event of $q$ again
  happens-after an event in an incomplete message on some other handler, this procedure is repeated recursively until convergence, resulting in an extension $w'$ of $w$.
  \item
  Thereafter, the happens-before relation $\happbf{hb}{E.w'}$ is extended to include ordering constraints induced by the event-driven execution model.
  \begin{enumerate}[(i)]
    \item First $\hbwpref{E}{w'}$ is constructed as
    the smallest transitive relation which includes $\happbf{\hb}{E.w'}$ and in addition enforces
    $e \hbwpref{E}{w'} e'$  whenever $e$ is in a message whose first event is in $E$ and $e'$  occurs after $e$ on the same handler as $e$.
    \item
   Thereafter, $\happbf{sc}{E.w'}$ is defined as the saturation $\weaksatrel{\hbwpref{E}{w'}}$ of $\hbwpref{E}{w'}$ (see \cref{def:weakall}).
  \end{enumerate}
  If now $e' \happbf{sc}{E.w'} e$ for some event $e$ in $p$ and event $e'$ in a message which precedes $p$ on the same handler, then $p \in \winits{\exseq}{w}$ must be false.
  \end{enumerate}
\item[Witness Construction]
  If the Happens-Before Check was not negative, the next step is to construct an actual execution in which $p$ is the first message on its handler. 
  First, $\happbf{sc}{E.w'}$ is extended, by ordering
  the events in $p$ before any event in a message that precedes $p$ in $w'$ on the same handler, and thereafter saturated by the saturation operation $\weaksatrel{\cdot}$.
  If the result contains a cycle, then $p \in \winits{\exseq}{w}$ must be false. Otherwise
  we extend the saturated extension of $\happbf{sc}{E.w'}$ to a total order on the messages of each handler, by ordering
  messages that are still unordered to execute in the same order as they appear in $w'$. If this can be done without creating a cycle then
  $p \in \winits{\exseq}{w}$ is true.
\item[Decision Procedure]
  If a cycle is created, then a decision procedure is invoked as a final step.
\end{description}

\section{Proofs of Complexity Results}
\label{sec:complexity-proofs}

In this section, we prove the complexity results of \cref{thm-lowerbound,thm-consistency} but first we need to define the  happens-before relation on the events of each execution sequence.

Given an execution sequence $E$,
we define the \emph{happens-before relation} on $E$, denoted $\happbf{\hb}{E}$,
as the irreflexive partial order on $\dom{E}$ induced by the union of three sub-relations,
$\happbf{\po}{E}$, $\happbf{\cnf}{E}$, and $\happbf{\pb}{E}$. 
Each of these is a sub-relation of $\totorder{E}$, defined as follows.
\begin{description}
\item
[$e \happbf{\po}{E} e'$] if $e$ and $e'$ are performed by the same message $p$.
\item
[$e \happbf{\cnf}{E} e'$] if $e$ and $e'$ access a common shared variable \texttt{x}, and
at least one of them writes to \texttt{x}.
\item
[$e \happbf{\pb}{E} e'$] if $\procof{e'}$ is the message that is posted by $e$ and $e'$ is the first event of $\procof{e'}$. 
\end{description}
Intuitively, $\happbf{\po}{E}$ (\emph{program order}) is the total order of events of each message.
Note that $\happbf{\po}{E}$  does not order events of different messages relative to each other.
The relation $\happbf{\cnf}{E}$ (\emph{conflicts with}) captures data flow constraints arising from reads and writes to shared variables.
The relation $\happbf{\pb}{E}$ (\emph{posted by}) captures the causal dependency from message posting to message execution.

\subsection{Proof of \cref{thm-lowerbound}}
\label{thm-lowerbound-weak}

%

We  prove the lower bound by reduction from the  VSC-read problem. The reduction is similar to the one from the  event-driven consistency problem to the VSC-read problem. The idea is to start  from an execution sequence and reversing the order of two messages will lead to the pattern used in the hardness  proof of the  event-driven consistency problem. In  this proof, we will  replace the conflict relation from   $p_e$ to $p_{e,x}$ by a sequence of conflict relations that go through two particular messages $p'_{e,x}$ and $p''_{e,x}$ if they are executed in a certain order. Otherwise, there is no conflict relation from $p_e$ to  $p_{e,x}$, and so the happens-before relation is acyclic.
 
We use the same set of assumptions as in the hardness proof of the  event-driven consistency problem. 
We now reduce the VSC-read problem to the order reversing problem.
Let $(S,\happbf{\hb}{S}=\happbf{\po}{S} \cup \happbf{\rf}{S})$ be a directed graph.
As in the previous proof, we associate a message $p_e$ for each event~$e$.
The message $p_e$ will be executed by the handler $h$. 
For every write event $e$ in $S$ executed by a thread $t$, we have a message $p_e$ 
of the form [\texttt{$x_e$ = 1};\, \texttt{$x_t$ = 1};\, \texttt{$x'_e$ = 1}].
For a read event $e'$ in $S$ executed by a thread $t$ and reading from the write event $e$,
we have a message $p_{e'}$ which is of the form
\mbox{[\texttt{a = $y_e$};\, \texttt{$x_t$ = 1};\, \texttt{$y_{e'}$ = 0}]}. 

We also use an extra handler $h_x$ for each variable $x$ (used in the events of $S$).
For each write event $e$ on the variable $x$, we have a message $p_{e,x}$ that will run the following sequence of statements \mbox{[\texttt{$x_e$ = 0};\, \texttt{$z_e$ = 0};\, \texttt{$y_e$ = 0};\, \texttt{$z'_e$ = 0}]}. We then add a conflict relation from the first write of $p_{e,x}$ to the first  write of $p_e$. This will force the message $p_e$ to start after $p_{e,x}$. For a read event $e'$  reading from the write event $e$,  we also add a conflict relation from  the third write of $p_{e,x}$  to the first  read of $p_{e'}$. Observe that we do not impose a direct conflict relation from  $p_{e}$ or $p_{e'}$ to $p_{e,x}$.

For each write event $e$ on the variable $x$, we have two messages $p'_{e,x}$ and $p''_{e,x}$ that run on a fresh handler $h_e$ the following sequence of statements 
[\texttt{a = z};\, \texttt{$x'_e$ = 1}] and [\texttt{$z_e$ = 0};\, \texttt{a = x}] respectively. We add a conflict relation from the second last event of  $p_{e}$ to the second event of $p'_{e,x}$ and from the first write event of
$p''_{e,x}$  to the last event of  the second write event of  $p_{e,x}$. Observe that in the case that $p'_{e,x}$ is executed before $p''_{e,x}$, we have an indirect conflict relation from the last write of $p_e$ to  the second write of $p_{e,x}$   through $p'_{e,x}$ and $p''_{e,x}$. In the case where we execute $p''_{e,x} $ before $p'_{e,x}$, there is no  (indirect) happens-before relation from $p_{e}$ to $p_{e,x}$.


In similar manner, for each read event $e'$ on $x$ in $S$ reading from the write event $e$, we have two messages $p'_{e',x}$ and $p''_{e',x}$ that run on a fresh handler $h_{e'}$ the following sequence of statements 
\mbox{[\texttt{a = z};\, \texttt{$y_{e'}$ = 1}]} and [\texttt{a = $z'_e$};\, \texttt{a = x}] respectively.  We add a conflict relation from the last event of $p_{e'}$ to the second event of $p'_{e',x}$ and from the first read event of  
$p''_{e',x}$  to the  last write event of  $p_{e,x}$. Observe that in the case that $p'_{e',x}$ is executed before $p''_{e',x}$, we have an indirect conflict relation from the last write of $p_{e'}$ to the last write of $p_{e,x}$ through $p'_{e',x}$ and $p''_{e',x}$.

To set the order of all $p'_{e,x}$ and  $p''_{e,x}$ ($p'_{e',x}$ and $p''_{e',x}$), we will use  two messages $p$ and $p'$ on a fresh handler $h'$ that run the following  statements [\texttt{$x_p$ = 1}; \texttt{z = 1}] and [\texttt{$x_{p'}$ = 1}; \texttt{x = 1}] respectively. We add then a conflict relation from the first read event of  $p'_{e,x}$ (resp. $p'_{e',x}$) to the  event of $p$ and from the write event of  
$p'$  to the last event of $p''_{e,x}$ (resp.  $p''_{e',x}$). Note that if $p$ is executed before $p'$ then  $p'_{e,x}$ (resp. $p'_{e',x}$) is executed before  $p''_{e,x}$ (resp. $p''_{e',x}$).

Let $(S',\happbf{\hb'}{S'}=\happbf{\po}{S'} \cup \happbf{\cnf}{S'}\cup  \happbf{\pb}{S'})$ be the constructed  \hb-trace from $(S,\happbf{\hb}{S}=\happbf{\po}{S} \cup \happbf{\rf}{S})$. It is easy to see that there is an execution sequence $E$ such that $(S',\happbf{\hb}{S'})$ is the  \emph{\hb-trace} of  $E$ and where the message $p'$  is executed before $p$ and   $p''_{e,x}$ (resp. $p''_{e',x}$) is executed before  $p'_{e,x}$ (resp. $p'_{e',x}$).

\begin{lemma}
\label{lemma2}
 There is an execution sequence $E'$ such that $(S',\happbf{\hb}{S'})$ is the  \emph{\hb-trace} of  $E'$ and where the first event of message $p$ is  the  first executed event in $E$    if and only if there is an execution $E''$ such that $S=\dom{E}$, $ \happbf{\po}{S} =\happbf{\po}{E} $ and 
 $\happbf{\rf}{S} \subseteq \happbf{\cnf}{E}$. 
\end{lemma}

Imposing $p$ to be executed before $p'$ will impose that every  $p'_{e,x}$ (resp. $p'_{e',x}$) is executed before  $p''_{e,x}$ (resp. $p''_{e',x}$) and so there will be an indirect relation from the last write of  (resp. $p_e$) $p_{e'}$ to the last (second) write of $p_{e,x}$ through $p'_{e,x}$ and $p''_{e,x}$ ($p'_{e',x}$ and $p''_{e',x}$). Thus, we are in similar case as in the hardness proof of the  event-driven consistency problem.
Furthermore, we have the first event of $E'$  can be the first event of $p$ since it is independent from any other event.

\begin{lemma} 
\label{lemma3}
$p \in \winits{\emptyseq}{E}$  if and only if there is an execution sequence $E'$ such that $(S',\happbf{\hb}{S'})$ is the  \emph{\hb-trace} of  $E'$ and where the message $p$ is  executed before $p'$.
\end{lemma}

Finally, \cref{thm-lowerbound} can be seen as an immediate corollary of \cref{lemma2,lemma3}.

  \subsection{Proof of \cref{thm-consistency}}
\label{sec:complexity-proof-consistency}
\paragraph{Upper-bound} Let  $(S,\happbf{\hb}{S})$ be a directed graph (i.e., \hb-trace) 
 where $S$ is a set of events and $\happbf{\hb}{S}=\happbf{\po}{S} \cup \happbf{\cnf}{S} \cup \happbf{\pb}{S}$ is a set of labeled edges.
To show that the event-driven consistency problem is NP, it suffices to first guess a total ordering $<_S$ between the messages on the same thread handler. Observe that we can have at most one incomplete message per handler which should be scheduled last with respect to $<_S$. We then use  the  total order relation $<_S$ to extend  the {\em program order} relation  $\happbf{\po}{S}$   into a total order relation $\happbf{\po}{}$ on the set of events executed by the same handler such that: $(1)$ $e \happbf{\po}{} e'$ if $e \happbf{\po}{S} e'$, and $(2)$
$e \happbf{\po}{} e'$ whenever $e$ and $e'$ are events in two different messages $p$ and $p'$ on the same handler and $p <_S p'$. Finally,  the extended  {happens-before relation} $\happbf{\hb}{}=\happbf{\po}{} \cup \happbf{\cnf}{S} \cup \happbf{\pb}{S}$  is acyclic ({which  is equivalent  to checking sequential consistency} of the extended graph $(S,\happbf{\hb}{})$) if and only if  there is an execution sequence $E$ such that $(S,\happbf{\hb}{})$ is the  \emph{\hb-trace} of  $E$ (i.e., $S=\dom{E}$, $ \happbf{\po}{} =\happbf{\po}{E} $,  $\happbf{\cnf}{S} = \happbf{\cnf}{E}$  and $\happbf{\pb}{S} = \happbf{\pb}{E}$). Observe that checking the acyclicity  of the extended {happens-before relation} $\happbf{\hb}{S}$ can be done in polynomial time. Furthermore, the execution sequence $E$ can be obtained   via the linearlization of the extended {happens-before relation} $\happbf{\hb}{}$  (since the extend  the {\em program order} relation $\happbf{\po}{}$  forces the messages on the same handler to be executed one after the other).

\paragraph{Lower-bound} We prove the lower bound by reduction from the problem of verifying the sequential consistency  of traces when  only the read-from relation is given. Hereafter, we call this problem VSC-read. The VSC-read problem consists in checking whether, given a directed graph $(S,\happbf{\hb}{S}=\happbf{\po}{S} \cup \happbf{\rf}{S})$
 where $S$ is a set of write and read events, $\happbf{\po}{S}$ is the program order relation that totally orders all the events of each thread, and $\happbf{\rf}{S}$ is the {\em read-from} relation that maps each read event to the  write event from which it gets its value, there is an execution sequence $E$ such that $S=\dom{E}$, $ \happbf{\po}{S} =\happbf{\po}{E} $ and 
 $\happbf{\rf}{S} \subseteq \happbf{\cnf}{E}$. 
The VSC-read problem is known to be  NP-complete in the size of the
program~\cite[Theorem~4.1]{DBLP:journals/siamcomp/GibbonsK97}.

We now reduce the VSC-read problem to the event-driven consistency problem. Let   $(S,\happbf{\hb}{S}=\happbf{\po}{S} \cup \happbf{\rf}{S})$ be a directed graph. 
To simplify the presentation\footnote{We assume that threads/messages are spawned/posted by a main thread, and that all shared variables get initialized to 0, also by the main thread. To make the presentation simple, we omit the events of the main thread.}, we assume w.l.o.g. that each write event is read by at least one read event. 
The main idea of our reduction is to associate a message $p_e$  for each event $e$. The message $p_e$ will be executed by the handler $h$. The order of the execution of these messages will correspond to a linearization  of the set of event $S$ (since all these messages will be executed by the same handler $h$). However, this poses a challenge  since such  reduction from the VSC-read problem to the event-driven consistency problem will  fix the order of write events on the same
variable (as it is implied by the conflict relation $\happbf{\cnf}{}$). To address this challenge, we  rename the shared variables used by each event in $S$ and thus there will be no conflict relation between write-write events (and therefore between read-write events too). However, this leads to a new challenge which is how to make sure that between a write event $e \in S$ and a read event $e' \in S$ that is reading from $e$ there is no other scheduled write event   in $S$ on the same variable between $e$ and $e'$. To address this second challenge, we use an extra handler $h_x$ per variable $x$ that executes a number of independent messages (one  per write event on x in $S$). The order in which these messages are executed corresponds to the order in which the write events  on the same variable are scheduled. Furthermore, we make sure that each read event is scheduled after the write event it reads from and before the next scheduled write event on the same variable.

Formally, for every write event $e$ in $S$ executed by a thread $t$, we create a message $p_e$ running on the thread handler $h$.
The message $p_e$ will be of the form [\texttt{$x_e$ = 1};\, \texttt{$x_t$ = 1};\, \texttt{$x_e$ = 1}].
For a read event $e'$ in $S$ executed by a thread $t$ and reading from the write event $e$, we create a message $p_{e'}$ running on the thread handler $h$.
The message $p_{e'}$ will be of the form [\texttt{a = $y_e$};\, \texttt{$x_t$ = 1};\, \texttt{a = $y_e$}]. We use the write event on $x_t$ to order the messages corresponding to events running on the same thread $t$ in $S$. In fact, we will simulate $\happbf{\po}{S}$ using $\happbf{\cnf}{S'}$ that will  totally order all the write events on $x_t$. This results in adding  a conflict relation $ \happbf{\cnf}{S'}$ between every two events corresponding to the writes on $x_t$ in two different messages $p_e$ and $p_{e'}$ iff $e \happbf{\po}{S} e'$. 

The  statements on $x_e$ and $y_e$ are used to force a total order on the messages corresponding to events on the same variable such that  all the read messages are scheduled just after their corresponding write messages. To that aim, we use an extra handler $h_x$ for each variable $x$ (used by the events of $S$). For each write event $e$ on the variable $x$, we create a message $p_{e,x}$ that will run the following sequence of statements [\texttt{$x_e$ = 0};\, \texttt{$x_e$ = 0};\, \texttt{$y_e$ = 0};\, \texttt{$y_e$ = 0}]. We then add a conflict relation from the first  write of $p_{e,x}$   to the first write of $p_e$ and from the last write of $p_e$ to the second write of $p_{e,x}$. This will force the message $p_e$ to start and end before its corresponding reads. For a read event $e'$  reading from the write event $e$,  we also add a conflict relation from the third write of $p_{e,x}$  to  the first  read of $p_{e'}$ and from the last read of $p_{e'}$ to  the last write of $p_{e,x}$. This conflict relation will force that the entire message $p_{e'}$ will be executed just after the message $p_e$ without the interleaving  of  any other message that corresponds to a write event on $x$ between $p_e$ and $p_e'$.

Observe that the messages $p_e$ are run one after the other (since they are on the same handler $h$). Furthermore, the constraints between the messages of handler $h$ and those of handler $h_x$ impose that the read message $p_{e'}$ is scheduled just after its corresponding write message $p_e$ but before the next scheduled write  message on the same variable.
Let $(S',\happbf{\hb'}{S'}=\happbf{\po}{S'} \cup \happbf{\cnf}{S'}\cup  \happbf{\pb}{S'})$ be the constructed  \hb-trace from $(S,\happbf{\hb}{S}=\happbf{\po}{S} \cup \happbf{\rf}{S})$. It is then easy to see that:

\begin{lemma}
 There is an execution sequence $E'$ such that $(S',\happbf{\hb}{S'})$ is the  \emph{\hb-trace} of  $E'$  if and only if there is an execution $E$ such that $S=\dom{E}$, $ \happbf{\po}{S} =\happbf{\po}{E} $ and 
 $\happbf{\rf}{S} \subseteq \happbf{\cnf}{E}$. 
\end{lemma}

\section{Proof of Correctness and Optimality}
\label{sec:correctness-proof}

In this section, we prove correctness (\cref{thm:correctness}) and optimality
(\cref{thm:optimality}) of the \EventDPOR algorithm.

\subsection{Proof of \cref{thm:correctness}}

Let us first prove \cref{thm:correctness}. This theorem follows from the more general
\cref{thm:correctness-general}, which we state and prove in this section.

Let us
assume a particular completed execution of \EventDPOR. This execution 
consists of a number of terminated calls to $\explore(E)$ for some values 
of the parameters $E$ and $\WuT$. Let $\exseqs$ denote the set of execution 
sequences $E$ that have been explored in some call $\explore(E)$. Define 
the ordering $\treeorder$ on $\exseqs$ by letting $E \treeorder E'$ if 
$\explore(E)$ returned before $\explore(E')$. Intuitively, if one 
were to draw an ordered tree that shows how the exploration has proceeded, then 
$\exseqs$ would be the set of nodes in the tree, and $\treeorder$ would be the 
post-order between nodes in that tree. \cref{thm:correctness} follows from the more
general \cref{thm:correctness-general}, stated here

\begin{theorem}[Correctness of \EventDPOR]
\label{thm:correctness-general}
Whenever a call to $\explore(E)$ returns during \cref{alg:eventdpor-access},
then for all maximal execution sequences $E.w$, the algorithm has explored
some execution sequence in~$\eqclass{E.w}$.
\end{theorem}

Since the initial call to the algorithm, $\explore(\emptyseq)$, starts with
the empty sequence as argument, \cref{thm:correctness-general}
implies that for all maximal execution sequences $E$ the algorithm
explores some execution sequence $E'$ which is in $\eqclass{E}$.
Note also that a sequence of form $E.w$ need not have been explored inside
the call $\explore(E)$, but can have been explored in some earlier call,
of form $\explore(E'.p)$ for some prefix $E'$ of $E$.

The proof of~\cref{thm:correctness-general} proceeds by induction on the set $\exseqs$ of execution sequences $E$ that
are explored during the considered execution, using the ordering $\treeorder$
(i.e., the order in which the corresponding calls to $\explore(E)$ return).

We first state and prove a small lemma.

\begin{lemma}
\label{lem:WuT-exseqs}
Let $\exseqs$ be the tree of explored execution sequences. and let $\treeorder$ be
the order in which the corresponding calls to $\explore(E)$ return.
Consider some point in the execution, and let  $\wut{\exseq}$ be the wakeup tree
at $\exseq$ at that point, for some $\exseq \in \exseqs$.
\begin{enumerate}
\item \label{lem:WuT-exseqs:1}
  If $w \in \wut{\exseq}$ for some $w$, then $\exseq.w \in \exseqs$.
\item \label{lem:WuT-exseqs:2}
  If $w \prec w'$ for $w,w' \in \wut{\exseq}$ then $\exseq.w \treeorder \exseq.w'$
\end{enumerate}
\end{lemma}
\begin{proof}
The lemma follows by noting how the exploration from any $\exseq
\in \exseqs$ is controlled by the wakeup tree $\wut{\exseq}$ at
\crefrange{algacsl:exploration-begin}{algacsl:event-call-explore}
of~\cref{alg:eventdpor-access}.
\end{proof}

We now continue with the proof of \cref{thm:correctness-general}.

\medskip
\noindent
{\sl Base Case:}
This case corresponds to the first execution sequence $E$ for which the call
$\explore(E)$ returns. By the algorithm, $E$ is already
maximal, so the theorem trivially holds.

\medskip
\noindent
{\sl Inductive Hypothesis:}
The theorem holds for all execution
sequences $E'$ with $E' \treeorder E$.

\medskip
\noindent
{\sl Inductive Step:}
Proof by contradiction. Let us assume that there exists an execution $E$
such that when the call to $\explore(E)$ returns,
there is a maximal execution sequence $E.w$ such that
\cref{alg:eventdpor-access} has not explored
any execution sequence in $\eqclass{E.w}$.
We will show that this leads to a contradiction.
So, let $E$ be the smallest such execution in the $\treeorder$ order.
Let $\done$ be the value of the mapping $\done$ when the call to $\explore(E)$ returns.
Note that for such $w$ to exist, $\exseq$ cannot be maximal, so $\done(E)$ contains at least one message.

For each message $p$ such that $p \in \done(E')$ for some $E'$ with $E' \leq E$, 
define $E_p'$ to be the longest such $E'$. Thus, if 
$p \in \done(E)$ then $E_p' = E$, otherwise if
$E_p'$ is defined it is a strict prefix of~$E$ with $p \in \done(E_p')$.
It follows that $E_p'.p \treeorder E$.
We further define $w_p'$ by $\exseq = E_p'.w_p'$.
For each message $p$ such that $E_p'$ is defined and $p \in \winits{E_p'}{w_p'}$, define
\begin{itemize}
\item $w_p$ as the longest prefix of $w$ such that $p \in \winits{E_p'}{w_p'.w_p}$ (such a prefix must exist since one candidate is the empty sequence),
\item $e_p$ as the first event in $w$ which is not in
  $w_p$.  Such an event $e_p$ must exist, otherwise $w_p = w$, which implies
  $p \in \winits{E_p'}{w_p'.w}$,
  which together with the Inductive Hypothesis contradicts the assumption that
  the algorithm has not explored
  any execution sequence in $\eqclass{E.w}$,
\item $w_p''$ as a sequence such that $w_p'.w_p\mtprefixafter{E_p'}{p.w_p''}$.
\end{itemize}
Among the messages $p$ for which $E_p'$ is defined and $p \in \winits{E_p'}{w_p'}$, select
$q$ such that $w_q$ is the longest prefix among $w_p$.
If there are several such messages $q$ with equally long prefixes $w_q$,
then among these pick $q$ such that
$E_q'.q$ is minimal with respect to $\treeorder$.
Let $w_q''$ be a sequence with $w_q'.w_q\mtprefixafter{E_q'}{q.w_q''}$. 

Let $p'$ be the message $\procof{e_q}$ of $e_q$.
We first note that $e_q$ must be a shared-variable access. To see why,
note that if $e_q$ would start the message $p'$, then no event of the message $p'$
can be in $w_q'.w_q$. Moreover, the handler of $p'$ must be free after $E_q'.w_q'.w_q$, and $E_q'.w_q'.w_q$ must contain the event which posts $p'$.
We can simply extend $w_q''$ until it starts message $p'$
and then we have a sequence $w_q'''$ with $w_q'.w_q.e_q\mtprefixafter{E_q'}{q.w_q'''}$, contradicting the
choice of $w_q$.

There are now two cases to consider.
\begin{enumerate}
\item
  $q$ does not start a message after $E_q'$. Then $E_q'$ contains the first part of message $q$ (up until but not including $\nextev{E_q'}{q}$).
  Since $w_q'.w_q\mtprefixafter{E_q'}{q.w_q''}$, it follows that $\nextev{E_q'}{q}$ does not conflict with any event in $w_q'.w_q$, and that
  we can choose $w_q''$ as $w_q'.w_q$. The only reason for $q \not\in \winits{E_q'}{w_q'.w_q.e_q}$ (which implies $w_q'.w_q.e_q \notmtprefixafter{E_q'}{q.w_q''.e_q}$)
  is that $\nextev{E_q'}{q}$ conflicts with $e_q$.
  This implies that the execution $E_q'.q.w_q'.w_q.e_q$ contains a race between $\nextev{E_q'}{q}$ and $e_q$.
  Let $w_q'''$ be $w_q'.w_q.e_q$ and let $E_q'.q.w_q'''.z$ be a maximal extension of $E_q'.q.w_q'''$.
  Then $\nextev{E_q'}{q} \revrace{E_q'.q.w_q'''.z} e_q$.
By the Inductive Hypothesis, $\explore(E_q'.q)$ has then explored some sequence $E_q'.q.z'$ in $\mtclass{E_q'.q.w_q'''.z}$.
When exploring it, the race $\nextev{E_q'}{q}\revmsgrace{E_q'.q.z'}{m} e_q$ between $\nextev{E_q'}{q}$ and $e_q$ will be detected (at \cref{algacsl:race-loop}). Then $\reverserace(E_q'.q.z',\nextev{E_q'}{q},e_q)$ will return maximal executions, one of which must contain $E_q'.w_q'.w_q'.e_q$ as a happens-before prefix.
\item
  $q$ starts a message after $E_q'$.
  Since $e_q$ is a shared-variable access, it can be simply added to  $p'$ in $w_q''$, obtaining $w_q'''$.
  Since $q \not\in \winits{E_q'}{w_q'.w_q.e_q}$, $w_q''$ must contain an event $e$, which is
  not in $w_q'.w_q$, which conflicts with $e_q$.
  This implies that the execution $E_q'.q.w_q'''$ contains a race between $e$ and $e_q$.
  Let $E_q'.q.w_q'''.z$ be a maximal extension of $E_q'.q.w_q'''$.
  Then $e \revrace{E_q'.q.w_q'''.z} e_q$.
By the Inductive Hypothesis, $\explore(E_q'.q)$ has then explored some sequence $E_q'.q.z'$ in $\mtclass{E_q'.q.w_q''.z}$.
When exploring it, the race $e \revmsgrace{E_q'.q.z'}{m} e_q$ between $e$ and $e_q$ will be detected (at \cref{algacsl:race-loop}). Then $\reverserace(E_q'.q.z',e,e_q)$ will return maximal executions, one of which must contain $E_q'.w_q'.w_q'.e_q$ as a happens-before prefix.
\end{enumerate}

Let $E_q'.w_q'.w_q'.e_q$ be reordered as $E_q'.v$.
It follows that
$q \not\in\winits{E_q'}{v}$, from the assumptions made when selecting $q$.
Moreover, there cannot be any $E'',w,p$ such that $E''.w = E_q'$ and $p \in \dom{\done(E'')}$ and $p \in \winits{E''}{w.v}$,
also by the assumptions made when selecting $q$.
Thus, the wakeup sequence $v$ will be inserted into the wakeup tree $\wut{E_q'}$ (\cref{algacsl:event-race-end}) by the call $\insertwus{v}{\exseq_q'}{\emptyseq}$.
We claim that this insertion will add a sequence of form $E.p$ with
$p \in \winits{\exseq}{w_q.\procof{e_q}}$. To see why, we consider the
definition of $\insertwus{v}{E_q'}{u}$ in \cref{alg:wakeuptree}.
We first claim that during the insertion, the sequence $u$ will always
satisfy $E_q.u \leq E$ and $v$ will satisfy $u'.w_q.\procof{e_q} \infirstseqs{E_q.u} v$,
where $u.u' = w_q'$.
This is trivially true initially. To see that
it is preserved by each iteration of the loop starting at \cref{algl:wut-foreach-child},
we consider the possible children of form $u.p$. Let $r$ be the message such that
$E_q'.u.r \leq E$ (if still $E_q'.u < E$).
We know that $E_q'.u.r$ is in $\exseqs$ when $\explore(\exseq)$ is returns.
Furthermore, for each branch $u.p$ with $E_q.u'.p \treeorder E_q.u'.r$ we have that
$p \not \in \winits{\exseq.u}{u'.w_q.\procof{e_q}}$ by the Inductive
Hypothesis and the assumption that $\exseq.w$ has not been explored. On the other hand
$r \in \winits{\exseq.u}{u'.w_q.\procof{e_q}}$, implying that either
$u.r$ is already in $\wut{\exseq_q'}$ during the insertion, in which case the
loop will move to the next iteration with invariants preserved, or
$u.r$ is not already in $\wut{\exseq_q'}$ in which case it must be added during the
current insertion and produce a branch $u.v$ such that
$u'.w_q.\procof{e_q} \infirstseqs{E_q.u} v$.
Thus, when the insertion of $v$ has completed, possibly after having been parked, the exploration tree will contain
an execution of form $E.v'$ with $w_q.\procof{e_q} \mtprefixafter{E} v'$, thereby
contradicting the assumption that $w_q$ is the longest extension of $E$ that has been explored.
This concludes the proof of the inductive step,
and~\cref{thm:correctness-general} is proven.
\qed

\subsection{Proof of \cref{thm:optimality}}

Let us next prove \cref{thm:optimality}. This theorem depends on \EventDPOR being able
to the following property P:
\begin{itemize}
   \item[P:]
whenever the exploration tree $\exseqs$ contains a node of form $\exseq.p$, then the algorithm will not add an execution of form $\exseq.w$ which is contained in
some execution of form $\exseq.p.w'$ for some $w'$, i.e., for which $p \in \winits{\exseq}{w}$.
\end{itemize}
If P is enforced, then~\cref{alg:eventdpor-access} cannot explore two equivalent maximal executions. To see this, let $\exseq$ be the longest common prefix of the two executions. Let the
execution of the two, which is explored first, be of form $\exseq.p.w'$. The other execution
will then be the continuation of a wakeup sequence, which is inserted as a new sequence $w$
from the node $\exseq$ in the exploration tree $\exseqs$, and thereafter extended to
$\exseq.w.v$. But if now $\exseq.p.w' \mtequiv \exseq.w.v$, then
$\exseq.w \mtprefix \exseq.p.w'$, which implies $p \in \winits{\exseq}{w}$, which contradicts
P.

It thus remains to check that property P is enforced. By inspection of
\cref{alg:eventdpor-access}, we see that whenever a new sequence is inserted into $\exseqs$,
which happens before inserting a new wakeup sequence (\cref{algacsl:event-test}),
inside procedure \insertwusname (\cref{alg:wakeuptree}) for wakeup tree insertion, and in
the procedure \insertpendingwuname (\cref{alg:pendingwus}) for inserting parked wakeup sequences.
Furthermore, for non-branching programs the test for $p \in \winits{\exseq}{w}$,
described in \cref{sec:checkwi}, is exact.
This concludes the proof of the theorem.
\qed


\section{Complete Set of Benchmark Tables} \label{app:eval-complete}
In this appendix, we include the complete set of benchmark results comparing
the performance of the \EventDPOR with that of the \OptimalDPOR algorithm,
with the LAPOR technique implemented in \GenMC and also with the baseline
algorithm of \GenMC which tracks the modification order (\genmcmo{\small}) of
shared variables. A subset of these results appears in the main body of the
paper.

\paragraph{Baseline Comparison}
First, we measure the performance of algorithm implementations on three
programs where all algorithms explore the same number of executions.  The
first two of them are simple programs where a number of threads post racing
messages to a \emph{single} event handler. Both programs are parametric on the
number of threads (and messages posted); the value of this parameter is shown
inside parentheses. The messages of the first program (\bench{writers})
consist of a store to the same atomic global variable followed by an assertion
that checks for the value written. The second program (\bench{posters}) is
similar but between the write and the assertion check the messages also post,
to the same handler, another message with an atomic store to the same global
variable; this increases the number of executions to examine.
Finally, the third program (\bench{2PC}) is a two-phase commit protocol used
by a coordinator and $n$ participant threads (i.e., $n+1$ handler threads in
total) to decide whether to commit or abort a transaction, by broadcasting and
receiving messages.

\begin{table}[t]
  \caption{Performance on programs where different DPOR algorithms implemented
    in \Nidhugg and \GenMC explore the same number of complete and blocked
    executions. Times (in seconds) show the relative speed of their
    implementations.}
  \label{tab:eval:baseline}
  \centering\SZ
  \pgfplotstablevertcat{\output}{results/laban/writers.txt}
  \pgfplotstablevertcat{\output}{results/laban/posters.txt}
  \pgfplotstablevertcat{\output}{results/laban/2PC.txt}
  \pgfplotstabletypeset[
    every row no 3/.style={before row=\midrule},
    every row no 6/.style={before row=\midrule},
  ]{\output}
\end{table}

Results from running these benchmarks for increasing number of threads are
shown in~\cref{tab:eval:baseline}. As can be seen, all algorithms explore the
same number of executions here. This allows us to establish that:
\begin{enumerate}[(i)]
\item \GenMC \genmcmo{\small} is fastest overall; in particular, it is $3$--$7$
  times faster than \Nidhugg \opt{\small} and about $8$--$9$ times faster
  than \Nidhugg \evt{\small}.
\item The overhead that LAPOR incurs over its baseline implementation in
  \GenMC is significant.
  Still, for the first two programs, which involve just one event handler and
  no blocked or aborted executions, \GenMC \lapormo{\small} beats \Nidhugg
  \evt{\small}.  However, \Nidhugg \evt{\small} is faster than \GenMC
  \lapormo{\small} on the third program (\bench{2PC}).
\item The overhead that \EventDPOR incurs over \OptimalDPOR for the extra
  machinery that its implementation requires is small but quite noticeable.
\end{enumerate}
The results from \bench{2PC} corroborate these conclusions. The blocked
executions in this benchmark are due to \emph{assume-blocking} and affect all
algorithms equally in terms of additional executions examined.  However,
notice that \GenMC \lapormo{\small} is affected more in terms of time overhead
compared to its baseline.

\paragraph{Performance on More Involved Event-Driven Programs}
The next two benchmarks were taken from a recent paper by Kragl et
al.~\citet{Kragl20}.
In \bench{buyers}, $n$ ``buyer'' threads coordinate the purchase of an item
from a ``seller'' as follows: one buyer requests a quote for the item from the
seller, then the buyers coordinate their individual contribution, and finally
if the contributions are enough to buy the item, the order is placed.
In \bench{ping-pong}, the ``pong'' handler thread receives messages with
increasing numbers from the ``ping'' thread, which are then acknowledged back
to the ``ping'' event handler.

\begin{table}[t]
  \caption{Performance on programs where different DPOR algorithms implemented
    in \Nidhugg and \GenMC examine the same number of traces, but LAPOR also
    explores a significant number of executions that need to be aborted.
    This negatively affects the runtime that SMC using LAPOR takes.}
  \label{tab:eval:whenblocked}
  \setlength{\tabcolsep}{2.5pt}
  \centering\SZ
  \pgfplotstablevertcat{\output}{results/laban/buyers.txt}
  \pgfplotstablevertcat{\output}{results/laban/ping_pong.txt}
  \pgfplotstabletypeset[
    every row no 3/.style={before row=\midrule},
  ]{\output}
\end{table}

Results from running these benchmarks are shown
in~\cref{tab:eval:whenblocked}.
In these two programs, all algorithms explore the same number of traces, but
LAPOR also explores a significant number of executions that cannot be
serialized and need to be aborted. This negatively affects the time that SMC
using LAPOR requires; \GenMC \lapormo{\small} becomes the slowest
configuration here.  In contrast, \Nidhugg \evt{\small} shows similar
scalability as baseline \GenMC and \Nidhugg \opt{\small}.

\paragraph{Performance on Event-Driven Programs Showing Complexity
  Differences Between DPOR Algorithms}
Finally, we evaluate all algorithms in programs where algorithms tailored to
event-driven programming, either natively (\EventDPOR) or which are
lock-aware (when handlers are implemented as locks), have an advantage.
We use six benchmarks.
The first (\bench{consensus}), again from the paper by Kragl et
al.~\citet{Kragl20}, is a simple \emph{broadcast consensus} protocol for $n$
nodes to agree on a common value. For each node~$i$, two threads are created:
one thread executes a \texttt{broadcast} method that sends the value of
node~$i$ to every other node, and the other thread is an event handler that
executes a \texttt{collect} method which receives~$n$ values and stores the
maximum as its decision. Since every node receives the values of all other
nodes, after the protocol finishes, all nodes have decided on the same value.
The second benchmark (\bench{db-cache}) is a key-value store system inspired
from Memcached, a well known distributed cache application. There are $n$
clients requesting a fixed sequence of storage accesses to a server via UDP
sockets (modeled as threads with mailboxes). On the server side there is
one worker thread per client to fulfill these requests. So multiple worker
threads on the server threads may race.
The third benchmark (\bench{prolific}) is synthetic: $n$ threads send
$n$ messages with an increasing number of stores to and loads from an atomic
global variable to one event handler.
The fourth benchmark (\bench{sparse-mat}) computes sparseness (number of
non-zero elements) of a sparse matrix of dimension $m \times n$. The work is
divided among $n$ tasks/messages and sent to different handlers, which then
compute and join these results.
The fifth benchmark (\bench{mat-mult}) implements concurrent matrix
multiplication taking two matrices of dimensions $m \times k$ and $k \times n$
as inputs. The work is divided among $n$ tasks/messages and sent to different
handlers, which then compute and join these results.
The last benchmark (\bench{plb}) is taken from a paper by Jhala and
Majumdar~\citet{popl07:JhalaM}. The main thread receives a fixed sequence of
task requests. Upon receiving a task, the main thread allocates a space in
memory and posts a message with the pointer to the allocated memory that will
be served by a thread in the future.


\begin{table}[t]
  \caption{Performance on programs that show complexity differences in the
    number of traces that different DPOR algorithms implemented in \Nidhugg
    and \GenMC explore.}
  \label{tab:eval:differences}
  \setlength{\tabcolsep}{1.5pt}
  \centering\SZ
  \pgfplotstablevertcat{\output}{results/laban/consensus.txt}
  \pgfplotstablevertcat{\output}{results/laban/db_cache.txt}
  \pgfplotstablevertcat{\output}{results/laban/prolific.txt}
  \pgfplotstablevertcat{\output}{results/laban/sparse-mat.txt}
  \pgfplotstablevertcat{\output}{results/laban/mat_mult.txt}
  \pgfplotstablevertcat{\output}{results/laban/plb.txt}
  \pgfplotstabletypeset[
    every row no 3/.style={before row=\midrule},
    every row no 6/.style={before row=\midrule},
    every row no 9/.style={before row=\midrule},
    every row no 12/.style={before row=\midrule},
    every row no 15/.style={before row=\midrule},
  ]{\output}
\end{table}

Results from running these six benchmarks are shown in~\cref{tab:eval:differences}.

In \bench{consensus}, all algorithms start with the same number of traces, but
\LAPOR and \EventDPOR need to explore fewer and fewer traces than the other
two algorithms, as the number of nodes (and threads) increases. Here too,
\LAPOR explores a significant number of executions that need to be aborted,
which hurts its time performance. On the other hand, \EventDPOR's handling
of events is optimal in this program, even though it is not non-branching.

The \bench{db-cache} program shows a case where \GenMC, both when running with
\genmcmo{\small} but also with \lapormo{\small}, is non-optimal.  In contrast,
\EventDPOR, even with employing the inexpensive redundancy checks, manages to
explore an optimal number of traces.

The \bench{prolific} program shows a case where algorithms not tailored to
events (or locks) explore $(n-1)!$ traces, while \LAPOR and \EventDPOR explore
only $2^n-2$ consistent executions, when running the benchmark with $n$ nodes.
We briefly explain why the number of feasible executions are $2^n-2$. Because
of the access patterns of global variables in this program, each message is
conflicting with the previous and the next messages. In an execution, these
conflicts can be represented by $n$ directed edges. So there are $2^n$
possible reorderings when both directions of each edge are considered. But two
of these reorderings are not possible because they create a cycle, hence the
$2^n-2$.
On this program, it can also be noted that \EventDPOR scales \emph{much}
better than \LAPOR here in terms of time, due to the extra work that \LAPOR
needs to perform in order to check consistency of executions (and abort some
of them).

The \bench{sparse-mat} program shows another case where algorithms that are
not tailored to events explore a large number of executions unnecessarily
(\timeout denotes timeout). This program also shows that \EventDPOR beats
\LAPOR time-wise even when \LAPOR does not explore executions that need to be
aborted.

Finally, \bench{plb} shows a case on which \EventDPOR and \LAPOR really
shine. These algorithms need to explore only one trace, independently of the
size of the matrices and messages exchanged, while DPOR algorithms not
tailored to event-driven programs explore a number of executions which
increases exponentially and fast.

\end{document}